\documentclass[letterpaper,11pt]{article}

\usepackage{fullpage}

\usepackage{amsmath}
\usepackage{amsfonts}
\usepackage{amsthm}

\usepackage{authblk}

\usepackage{subfig}

\newtheorem{theorem}{Theorem}
\newtheorem{lemma}[theorem]{Lemma}
\newtheorem{definition}[theorem]{Definition}
\newtheorem{corollary}[theorem]{Corollary}
\newtheorem{proposition}[theorem]{Proposition}

\newtheorem{observation}[theorem]{Observation}

\newcommand{\eps}{\varepsilon}
\newcommand{\gdy}{\textsc{greedy}}
\newcommand{\alg}{\textsc{alg}}
\newcommand{\opt}{\textsc{opt}}

\newcommand{\cP}{\mathcal{P}}
\newcommand{\minD}{D^{\min}}
\newcommand{\omb}{\textsc{ob}}
\newcommand{\In}{\mathcal{I}_n}
\newcommand{\kCtr}{\textsc{k-Center}}
\newcommand{\Cls}{\mathbb{C}}

\newcommand{\lG}{1.25}

\usepackage{tikz}
\usetikzlibrary{calc}

\usetikzlibrary{decorations.pathreplacing}

\usepackage{color}
\definecolor{darkred}{rgb}{0.5,0,0}

\newenvironment{packed_item}{
\begin{itemize}
  \setlength{\itemsep}{1pt}
  \setlength{\parskip}{0pt}
  \setlength{\parsep}{0pt}
}{\end{itemize}}

\title{Stochastic dominance and the bijective ratio of online algorithms\footnote{This work was supported by project ANR-11-BS02-0015 
``New Techniques in Online Computation'' (NeTOC). The second author is supported in part by the European Research Council  under the European Union's Horizon 2020 research and innovation programme (grant agreement No 648032).}}

\author[1]{Spyros Angelopoulos}
\author[2]{Marc P. Renault}
\author[3]{Pascal Schweitzer}
\affil[1]{CNRS and Universit\'e Pierre et Marie Curie,  Paris, France\\ {\tt spyros.angelopoulos@lip6.fr}}
\affil[2]{CNRS and Universit\'e Paris Diderot, Paris, France\\ {\tt mrenault@liafa.univ-paris-diderot.fr}}
\affil[3]{RTWH Aachen University, Aachen, Germany\\ {\tt schweitzer@informatik.rwth-aachen.de}}

\date{}

\begin{document}

\maketitle

\begin{abstract}
{\em Stochastic dominance} is a technique for evaluating the performance of online algorithms that provides an intuitive, yet powerful stochastic order between the compared algorithms. Accordingly this holds for {\em bijective analysis}, which can be interpreted as stochastic dominance assuming the uniform distribution over requests. These techniques have been applied in problems such as paging, list update, bin coloring, routing in array mesh networks, and in connection with Bloom filters, and have provided a clear separation between algorithms whose performance varies significantly in practice. However, despite their appealing properties, there are situations in which they are not readily applicable. This is due to the fact that they stipulate a stringent relation between the compared algorithms that may be either too difficult to establish analytically, or worse, may not even exist.

In this paper we propose remedies to both of these shortcomings. First, we establish sufficient conditions that allow us to prove the bijective optimality of a certain class of algorithms for a wide range of problems; we demonstrate this approach in the context of well-studied online problems such as weighted paging, reordering buffer management, and 2-server on the circle. Second, to account for situations in which two algorithms are incomparable or there is no clear optimum, we introduce the {\em bijective ratio} as a natural extension of (exact) bijective analysis. Our definition readily generalizes to stochastic dominance. This renders the concept of bijective analysis (and that of stochastic dominance) applicable to all online problems, is a broad generalization of the Max/Max ratio due to Ben-David and Borodin, and allows for the incorporation of other useful techniques such as amortized analysis. We demonstrate the applicability of the bijective ratio to one of the fundamental online problems, namely the continuous k-server problem on metrics such as the line, the circle, and the star. Among other results, we show that the greedy algorithm attains bijective ratios of $O(k)$ consistently across these metrics. These results confirm extensive previous studies that gave evidence of the efficiency of this algorithm on said metrics in practice, which, however, is not reflected in competitive analysis.
\end{abstract}

\newpage

\section{Introduction}
\label{sec:introduction}

Competitive analysis provides a simple yet effective framework for evaluating the performance
of online algorithms. Given a cost-minimization problem, the 
{\em competitive ratio} of the online algorithm $A$ is defined as $\sup_\sigma \frac{A(\sigma)}{\opt(\sigma)}$, 
where $A(\sigma)$ and $\opt(\sigma)$
denote the cost of $A$ and the optimal cost on a request sequence $\sigma$, respectively.  
This concept of comparing the worst-case performance of an online algorithm (with no advance knowledge of the sequence) 
to an optimal solution (with full access to the sequence) was first used by Graham in 1966~\cite{Graham1966} to analyze algorithms for the job shop scheduling problem. 
Following the seminal work of Sleator and Tarjan in 1985~\cite{ST85}, competitive analysis became the standard yardstick in the evaluation of online algorithms, and it  
has been instrumental in shaping
online computing into a well-established field of theoretical computer science. 
Overall, competitive analysis is broadly applicable and gives valuable insight into the performance of online algorithms.

Notwithstanding the undeniable success of competitive analysis, certain drawbacks have long been known. 
Most notably, due to its pessimistic (i.e., worst-case) evaluation of algorithms, it often fails to  
distinguish between algorithms for which experimental evidence (or even plain intuition)
would show significant differences in terms of performance. 
One definitive illustration of this undesirable 
situation is the well-known {\em paging problem}; here, paging strategies that are very efficient in practice, such as Least-Recently-Used 
have the same competitive ratio as extremely naive and costly strategies, such as Flush-When-Full~\cite{ST85}. 
Generally, competitive analysis is particularly meaningful when the obtained ratios are small; however, when the ratios are large 
(and even more so, in the case of an unbounded ratio) it risks no longer reflecting what is observed in practice.
Such disconnects between the empirical and the theoretical performance evaluation have motivated a substantial line of research on measures 
alternative to the competitive ratio. Some of the known approaches are: 
the {\em Max-Max ratio}~\cite{BenBor94:maxmax}; 
the {\em diffuse adversary model}~\cite{diffuse,diffuse1,diffuse2}; 
{\em loose competitiveness}~\cite{loose1,caching1}; 
the {\em random order ratio}~\cite{kenyon96bestfit};
the {\em relative worst-order ratio}~\cite{BFL05:relative,BEL06:LRU2}; 
the {\em accommodation function model}~\cite{BLN01:accommodating}; and 
{\em stochastic dominance} as well as {\em bijective and average analysis}~\cite{DBLP:journals/ipl/Hofri83,MR812817,MR793474,DBLP:journals/ipl/SeshadriR96,DBLP:journals/jcss/Mitzenmacher96,DBLP:journals/im/LumettaM07,ADLO07:paging,DBLP:conf/esa/HillerV08,ADLO08:list,AS:bijective}.
We refer the reader to the surveys~\cite{survey,DBLP:journals/ife/HillerV12} for an in-depth discussion of such techniques. 

Of particular interest in this work is (first order) \emph{stochastic dominance} which defines a partial order on random variables. A random variable~$X$ is 
stochastically dominated by a random variable~$Y$ if, for all~$c\in\mathbb{R}$, we have~$\Pr[X\leq c]\geq \Pr[Y\leq c]$. Clearly, if $X$ is stochastically dominated by $Y$,
then $\mathbb{E}(X) \leq \mathbb{E}(Y)$; however a much stronger conclusion can be drawn, namely that, for the cumulative distribution functions of the distributions from which $X$ and $Y$ are drawn, denoted by $F$ and $G$ respectively, $F(c) \ge G(c)$ for all $c$. That is, informally, $F$ has more probability mass towards lower values than $G$.
Moreover, it can be shown that $\mathbb{E}(h(X)) \leq \mathbb{E}(h(Y))$ for every increasing function $h$. If we think of $h$ as a utility function, then stochastic dominance provides a kind of {\em unanimity rule} which informally states that $X$ should be preferred
to $Y$ under any monotone utility function (assuming these variables denote costs). For this reason, stochastic dominance 
has been very useful in the context of decision theory and microeconomics, with applications varying from portfolio selection to measuring income inequality in society.
For a comprehensive discussion, see Chapters 4 and 5 in the textbook~\cite{Wolfstetter}.  

Hiller and Vredeveld~\cite{DBLP:conf/esa/HillerV08} applied the concept of stochastic dominance in the context of online computing.
More precisely, an algorithm~$A$ is \emph{stochastically no worse} 
than an algorithm~$B$ with respect to a given distribution over the request sequences
if the random variable corresponding to the cost of~$A$ is stochastically dominated by that of~$B$. 
In particular, assuming the uniform distribution over all request sequences of a given size,  
stochastic dominance is equivalent to {\em bijective analysis}~\cite{DBLP:journals/ife/HillerV12}. 
This latter notion was first introduced in~\cite{ADLO07:paging} in the context of the paging problem and was shown to be
consistent with some natural, ``to-be-expected'' properties of efficient online algorithms 
(e.g., the effect of locality of reference as well as lookahead) which competitive analysis fails to yield. 
For a further discussion of the appealing aspects of bijective analysis, see~\cite{ADLO07:paging,AS:bijective}.

\begin{definition}[\cite{ADLO07:paging}] 
Let ${\cal I}_n$ denote the set of all request sequences of size $n$.
The online algorithm $A$ is {\em no worse} than 
the online algorithm $B$ on inputs of size $n$ according to bijective analysis
if there exists a bijection $\pi: {\cal I}_n \rightarrow {\cal I}_n$
satisfying $A(\sigma) \leq B(\pi(\sigma))$ for each $\sigma \in {\cal I}_n$. 
Moreover, $A$ is bijectively {\em optimal} if the above holds for all online algorithms $B$.
\label{def:bijective}
\end{definition}

\paragraph{The bijective ratio of online algorithms.}
Despite the appealing properties of bijective analysis, its biggest deficiency is a rather serious one: 
given two online algorithms it may be very difficult to compare them, in that it may be very hard to prove analytically the existence of the required bijection; 
even worse, such a bijection may not even exist. Thus, this analysis technique may deem
algorithms incomparable across a wide variety of problems, and, in this sense, it does not give rise to a real performance measure. 
This drawback implies that bijective analysis (and by extension, stochastic dominance more generally) 
lacks the most desirable property of the competitive ratio; namely the amenability of any given online problem to analysis.
Such an observation could also help explain why these techniques did not become as popular as competitive analysis, even though the fundamentals and 
some limited applications can be traced to work contemporary of competitive 
analysis. 
Calderbank et al.~\cite{MR812817} point towards such difficulties when observing in the context of the $k$-server problem that
``The prospects for successful analysis would seem to be better for the circle.
However, even in this case optimization questions may well be intractable since rules
of the simplicity of the [greedy]  rule are unlikely to be optimal''
(see also the discussion in Section~\ref{sec:preliminaries}).
For this reason, \cite{ADLO07:paging} introduced a substantially weaker technique termed {\em average analysis} 
that compares the average cost of two algorithms over requests of the same length. 
In particular,  we say that $A$ is {\em no worse} than $B$ on inputs of size $n$ according to average analysis
if $\sum_{\sigma \in {\cal I}_n} A(\sigma) \leq \sum_{\sigma \in {\cal I}_n} B(\sigma)$. 
Note that, if $A$ is no worse than $B$ according to bijective analysis, then the same relation holds for average analysis; however the opposite is not necessarily true.

In this paper, we propose (and apply) an extension of bijective analysis that makes the technique applicable to any given online problem: this extension gives rise to a performance measure which we call the {\em bijective ratio}. 
This is equivalent to an approximate stochastic dominance under a uniform distribution and can be readily generalized to the \emph{stochastic dominance ratio} for any distribution. 
\begin{definition}
Given an online algorithm $A$ and an algorithm $B$, and $n \in {\mathbb N^+}$, we say that 
the {\em bijective ratio of $A$ against $B$} is at most $\rho$ if there exists a bijection $\pi: {\cal I}_n \rightarrow {\cal I}_n$ satisfying 
$A(\sigma) \leq \rho \cdot B(\pi(\sigma))$ for all $n \geq n_0$. We denote this by $A \preceq_b \rho \cdot B$. 
The bijective ratio of an online 
algorithm $A$ is at most $\rho$ if, for every algorithm $B$, the bijective ratio of $A$ against $B$
is at most $\rho$. The bijective ratio of an online cost-minimization problem is the minimum $\rho$ 
for which there exists an online algorithm with bijective ratio at most $\rho$. 
\label{def:bijective.ratio}
\end{definition} 
We note that, in Definition~\ref{def:bijective.ratio}, one may allow $B$ to be either an online, or an offline algorithm (and in particular, the offline
optimum). We can thus distinguish between the bijective ratio of an online algorithm {\em against online or offline algorithms}.
We clarify that unless explicitly specified,``the bijective ratio of an algorithm $A$'' assumes a comparison
against the offline optimum. 

The above distinction is motivated similarly to the Max/Max ratio introduced by Ben-David and Borodin~\cite{BenBor94:maxmax}, which is defined as
the ratio of the maximum-cost 
sequence for algorithm $A$ over the maximum-cost sequence for algorithm $B$ (which may be online or offline), 
for a given sequence length. 
We emphasize that the bijective ratio is a strong generalization of the Max/Max ratio; namely it
implies that 
the cost of the $i$-th most expensive sequence of $A$ is at most $\rho$ times the cost of the $i$-th most expensive 
sequence of $B$ {\em for all $i$}, and not just for the most expensive sequences of $A$ and $B$. 

Definition~\ref{def:bijective.ratio} is a natural extension of bijective optimality in the spirit of measures such as the competitive and the approximation
ratio. It also upholds the essential aspect of bijective analysis in that every sequence for which $A$ incurs a certain cost 
can be bijectively mapped to a sequence on which $B$ is at most $\rho$ times as costly. 
Furthermore, a bijective ratio of $\rho$ implies that the average-cost ratio of the two algorithms is at most $\rho$, 
but also the far stronger conclusion that the contribution of sequences to the average costs of the two
algorithms can be attributed in a local manner, as argued above.
This aspect extends to any stochastic dominance ratio for any distribution.
Both properties are desired extensions of bijective optimality, in the sense that they provide
a much stronger comparison than the one induced by average-case analysis.

Last, we note that in the above definitions, the performance ratios are {\em strict}; however, as with the competitive ratio, one can easily define {\em asymptotic} ratios. For instance, the asymptotic ratio of $A$ against $B$ 
is at most $\rho$ if there exists a constant $c$ such that  $A(\sigma) \leq \rho \cdot B(\pi(\sigma))+c$ for all $n \geq n_0$. 

\paragraph{Contribution.} 
Our main objective is to expand the applicability of stochastic dominance and, more  specifically, bijective analysis. We accomplish this in two ways:
first, by giving general, sufficient conditions for bijectively optimal algorithms; second, by applying 
the measure of bijective ratio to one of the canonical online problems, namely the $k$-server problem, as a case study. 

We begin our study of bijective analysis in Section~\ref{sec:optimality}, in which we extend the techniques 
applied in~\cite{AS:bijective} so as to prove the bijective optimality of certain types of greedy algorithms for a much wider class of 
problems than paging and list update. In particular, we identify some essential conditions under which a  certain subclass of greedy-like algorithms
(as formally defined in~\cite{DBLP:journals/algorithmica/BorodinNR03}) are optimal.
We then apply this general framework to  
the {\em 2-server problem on the continuous circle}, the {\em weighted paging} problem, and the {\em reordering buffer management} problem.
The above are all widely studied problems in online computing.
In particular, the result for the 2-server problem on the continuous circle  
improves on the result of Calderbank et al.~\cite{MR812817} which holds for the average case only. We also note that 
Anagnostopoulos et al.~\cite{anagnostopoulos:stochastic.k.server} studied the steady state of a stochastic version of the $k$-server problem on the circle, 
and reproved the optimality, in the average case, of the greedy algorithm when $k=2$. 

Our second contribution addresses the situation in which, according to stochastic dominance or bijective analysis, optimal algorithms may not necessarily exist.    
More precisely, we demonstrate the applicability of the bijective ratio 
in the analysis of the continuous $k$-server problem on the line, circle, and star metrics. 
Our main focus is on the performance of the greedy algorithm which is motivated by several factors. First, and most importantly, there is 
ample experimental evidence that in practice the greedy algorithm performs well in several settings~\cite{MR812817,DBLP:journals/cit/BaumgartnerMH07,DBLP:journals/cit/RudecBM10,DBLP:journals/cejor/RudecBM13}. However, these results are in stark contrast with competitive analysis since the greedy algorithm
has an unbounded competitive ratio even on the line.  
As noted in~\cite{DBLP:journals/cit/RudecBM10}, ``the [experimental] results demonstrate that [Work Function Algorithm (\textsc{wfa})] performs similarly or
only slightly better than a simple heuristic such
as the greedy algorithm, although according to
theory it should perform much better''.
In this sense, there is a big disconnect between theoretical and practical
behaviour which, perhaps surprisingly, has not received as much attention from the theoretical computer science community as other problems such as the paging problem. 
Our results demonstrate that bijective analysis can help bridge this gap. 
More precisely, we show that the greedy algorithm has bijective ratio $O(k)$ in the considered metric spaces.
Note that, for the $k$-server problem on the circle, we obtain a bijective ratio of
$k$, while the best-known competitive ratio is $2k-1$ by the analysis of the \textsc{wfa}~\cite{Borodin:textbook}.

Another appealing property of the greedy algorithm for the $k$-server problem, which is also true for the other online problems we study, is that they are among the 
simplest {\em memoryless} algorithms one can devise. Memoryless algorithms are very desirable, in general, and 
are particularly important for paging problems \cite{DBLP:journals/ibmrd/RaghavanS94}, and controlling disk heads \cite{MR812817}. 
In the context of the $k$-server problem, in particular, it is known that \textsc{wfa} is prohibitive in practice as it requires a full history of the requests and is much more complicated to implement than the simple greedy algorithm \cite{DBLP:journals/cit/BaumgartnerMH07,DBLP:journals/cit/RudecBM10,DBLP:journals/cejor/RudecBM13}.
Our result on the bijective optimality of a greedy policy concerning the weighted paging problem is of note given 
that~\cite{DBLP:journals/siamdm/ChrobakKPV91} showed that no deterministic memoryless algorithm has bounded competitive ratio for this problem.

In Section~\ref{sec:lineandcircle}, we first show that the greedy algorithm (denoted by \gdy) 
is not an optimal online algorithm for 2-server on the line, even for average-case analysis
(which implies the same result for bijective analysis). This improves on a result of Calderbank et al.~\cite{MR812817} that 
showed that there exists a {\em semi-online} algorithm (that knows the length of the sequence) that outperforms \gdy \  only on the last two requests. 
We also show that no online algorithm has a strict bijective ratio better than 2 for this problem. 
This immediately raises the question: How good (or bad) is \gdy? We address this question by 
showing that \gdy \ has a strict bijective ratio of at most $k$  and $2k$ for the circle and the line,
respectively; for the line, we also obtain an asymptotic bijective ratio at most $4k/3$.  
This analysis is almost tight, since we show that the asymptotic bijective ratio of \gdy \ is at least 
$k/3-\epsilon$ and $k/2-\epsilon$, for the circle and the line, respectively.
We also consider the algorithm $\kCtr$~\cite{BenBor94:maxmax}, which anchors its servers at $k$ points of the metric so as to minimize the maximum distance of any point in the metric to a server; it then serves each request by moving the closest server which subsequently returns to its anchor position.
In contrast to the results for $\gdy$, $\kCtr$ has an asymptotic bijective ratio of 2 for the line and the circle which generalizes the known bound on the Max/Max ratio of this algorithm~\cite{BenBor94:maxmax}.
In terms of a direct comparison of online algorithms, we obtain that \gdy \ has a bijective ratio of at most $2k/3$ against \kCtr. 
It is  worth mentioning that 
our results expand the work of Boyar et al.~\cite{Boyar:measures} who showed that $\gdy$ is bijectively optimal for the 2-server problem on a very simple, albeit discrete metric consisting of three colinear points (termed the {\em baby server problem}). 

Last, in Section~\ref{sec:star} we consider the continuous $k$-server problem on star-like metrics. Here, we show that the bijective ratio of 
the greedy algorithm is at most $4k$. On the negative side, we show that \kCtr \ is unbounded for such metrics. This raises
an interesting contrast between the bijective ratio and the Max/Max ratio: while \kCtr \  has Max/Max ratio at most $2k$ for $k$-server
on any bounded metric space, when considering the bijective ratio (which, as noted earlier, generalizes the Max/Max ratio), this algorithm 
becomes very inefficient.

In terms of techniques, the transition from exact to approximate bijective analysis 
necessitates a new approach that combines bijective analysis and amortization arguments.
In particular, we note that all previous work that establishes the bijective optimality of 
a given algorithm~\cite{ADLO07:paging,ADLO08:list,AS:bijective,Boyar:measures} is based on inductive arguments 
which do not immediately carry over to the bijective ratio. For instance~\cite{AS:bijective} 
crucially exploits the fact that for $\rho=1$, if $A \preceq_{b} \rho \cdot B$ and $B \preceq_{b} \rho \cdot C$,
then $A \preceq_b C$. This obviously only holds for $\rho=1$, i.e., for optimality.  
We thus follow a different approach that is based on a {\em decoupling} of the costs incurred by the compared algorithms
(stated formally in 
Lemmas~\ref{lem:genApproach} and~\ref{lem:amort}) by formulating two desirable properties: the first property captures the ``local'' 
efficiency of the greedy algorithm (but also potentially other good algorithms), while the second property allows
us to define best and worst server configurations (or approximations thereof) that provide insights into the choice of the appropriate
bijection. Combining these properties yields the desired results. 
For line and star metrics, in particular, 
we resort to amortized analysis using explicit potential functions which is the first example of a combination
of bijective and amortized analysis.

We conclude this section with two observations. First, we use the bijective ratio both 
to compare algorithms against the optimal offline algorithm (similar to the competitive ratio) and to directly and indirectly compare online algorithms. As an example of indirect comparison, 
Theorem~\ref{thm:kcenter.2} implies that \kCtr \ has a bijective ratio of at most 2 against \gdy, which, in combination with 
Theorem~\ref{lem:linek2} implies that \gdy \ has a bijective ratio of $\Omega(k)$ against $\kCtr$ for the line and the circle. 
The latter is asymptotically tight due to a direct-comparison result (stated in Theorem~\ref{thm:circleK}).
Second, while our focus is mainly on the greedy algorithm for reasons argued earlier, our techniques 
(in particular the decoupling Lemmas~\ref{lem:genApproach} and \ref{lem:amort}) are not tied to \gdy \ or \kCtr, and are potentially applicable to a wider class of algorithms.

\section{Related work and preliminaries}
\label{sec:preliminaries}

\paragraph*{Related work.}
Stochastic dominance (cf.\ \cite{ShakedBook,MullerBook,Wolfstetter,DBLP:journals/ife/HillerV12}) is a widely established concept in decision theory. Optimal algorithms, 
assuming certain pertinent distributions, have been identified for various online problems such as the paging problem~\cite{ADLO07:paging,AS:bijective,HV:2009.tech.report}, the list update problem~\cite{ADLO08:list,AS:bijective}, routing in array mesh networks~\cite{DBLP:journals/jcss/Mitzenmacher96}, bin colouring~\cite{DBLP:conf/esa/HillerV08} and in the online construction of Bloom filters~\cite{DBLP:journals/im/LumettaM07}. The first application of stochastic dominance for the analysis of online algorithms 
can be traced back to~\cite{DBLP:journals/ipl/Hofri83,DBLP:journals/ipl/SeshadriR96} in the context of the {\em two-headed disk} problem. 
This problem is related to the $k$-server problem but with a different cost function. Given $k$ mobile servers on a metric space, request appear on the points of the metric space and the goal is to minimize the distance travelled to serve these requests. During the time a request is being served, the other servers can re-position themselves at no cost. This renders the decision of which server to use trivial (it will always be the closer server) and puts focus on the question of where to place the other servers. 
Hofri showed that the natural greedy algorithm for this problem on the line is optimal in average and conjectured that 
it is stochastically dominated by every other algorithm under a uniform distribution~\cite{DBLP:journals/ipl/Hofri83} which
was proven by Seshadri and Rotem~\cite{DBLP:journals/ipl/SeshadriR96}.

The $k$-server problem
, originally proposed by Manasse et al.~\cite{MMS88},
involves $k$ mobile servers over a metric space. 
Upon a request to a node, a server must be moved to it. 
The incurred cost is the overall distance traversed by the servers. 
The $k$-server
problem generalizes the paging problem and has motivated an outstanding body of research (see the
surveys~\cite{ChrobakL92} and~\cite{DBLP:journals/csr/Koutsoupias09}). 
In general metric spaces, 
the Work Function Algorithm (\textsc{wfa}) of Koutsoupias and Papadimitriou is $(2k-1)$-competitive~\cite{KP95};
the best-known lower bound on the deterministic competitive ratio is $k$~\cite{MMS88}. \textsc{wfa} is 
also known to be $k$-competitive on the line~\cite{BartalK04} as well as for two servers in 
general metric spaces~\cite{ChrobakL92}, and it is the best known algorithm for the circle~\cite{Borodin:textbook}. Chrobak and Larmore showed that the algorithm Double Coverage,
which moves certain servers at the same speed in the direction of the request until a server reaches the requested point, 
is  $k$-competitive for the tree metric~\cite{Chrobak91tree}.
Calderbank et al.\ studied the $2$-server problem on the line and circle~\cite{MR812817}, and the $n$-dimensional sphere~\cite{MR793474}. 
They focused on the average case and, in particular, calculated the expected cost of $\gdy$ on the circle. 
Moreover, \cite{MR812817} presents experimental data that show that $\gdy$ is relatively close in performance to the offline optimal algorithm on the line. Similar experiments, for a variety of metric spaces and algorithms, including \gdy, are presented 
in~\cite{DBLP:journals/cit/BaumgartnerMH07,DBLP:journals/cit/RudecBM10,DBLP:journals/cejor/RudecBM13}. 
In a related work, Anagnostopoulos et al.~\cite{anagnostopoulos:stochastic.k.server} studied the steady-state distribution of $\gdy$ for the $k$-server problem on the circle.

Boyar et al.~\cite{Boyar:measures} provided a systematic study of several 
measures for a simple version of the $k$-server problem, namely, the two server problem on three colinear points. In particular, 
they showed that $\gdy$ is bijectively optimal.
Concerning the Max/Max ratio,~\cite{BenBor94:maxmax} showed that the algorithm \kCtr \
is asymptotically optimal up to a factor of 2 among all online algorithms and up to a factor of $2k$ from the optimal offline algorithm.

\paragraph{Preliminaries.}
We denote by $\sigma$ a sequence of requests, and by 
$\In$ the set of all request sequences of size $n$. Following~\cite{AS:bijective}, we denote by $\sigma[i,j]$ 
the subsequence $\sigma[i] \ldots \sigma[j]$.
We also use sometimes $\sigma_i$ to refer to the $i$-th request of $\sigma$, namely $\sigma[i]$. 
For the $k$-server problem, we denote the 
distance between two points $x,y$ by $d(x,y)$. Unless otherwise noted, we assume that both the line and the circle have unit lengths. 

Since, for continuous 
metrics, $\In$ is infinite, one needs to be careful about the allowable bijections. 
We model the continuous $k$-server problem using discrete 
metrics in which nodes are placed in an equispaced manner; as the number of nodes approaches infinity, 
this model provides a satisfactory approximation of the continuous problem. 
For instance, we approximate the continuous line (resp. circle) by a path (resp. cycle) in which vertices are uniformly spaced,
i.e., all edges have the same length which may be arbitrarily close to zero.
However, we note that the techniques we use in this paper are applicable even for the formal definition of the 
continuous problem, i.e., even when the set of all request sequence of size $n$ is infinite. 
However, in this case one needs to be careful about the allowable bijections. For instance, we should not
allow bijections that map the unit line to segments of measure strictly smaller than one. For this 
reason, we restrict the allowable bijections to {\em interval exchange transformations}~\cite{keane:1975}. 
These transformations induce bijections of the continuous space $[0,1]$ to itself that preserve 
the Lebesgue measure. Note that an interval exchange transformation is continuous with the exception 
of a finite number of points. In particular, we apply such transformations when constructing the bijection 
request-by-request.

Given an online algorithm $A$, we say that the \emph{configuration} of $A$ after serving any sequence of requests $\sigma$ 
is the state of the algorithm immediately after serving $\sigma$, where the notion of ``state'' will be implicit in the definition 
of the online problem. For example, in the $k$-server problem, this would be the position of the servers in the metric space.

\section{A sufficient condition for optimality of greedy-like algorithms}
\label{sec:optimality}
In this section, we show how the techniques of~\cite{AS:bijective} can be applied in a variety of online problems,
so as to prove that certain greedy algorithms are bijectively optimal among all online algorithms. 
To this end,
we first need a criterion that establishes, in a formal manner, the greedy characteristic. 
More precisely, consider an online algorithm that must serve request $\sigma_i$ after having served the sequence $\sigma[1,i-1]$. 
We say that an algorithm is {\em greedy-like} if it serves each request $\sigma_i$ in a way that minimizes the cost objective, 
assuming this request $\sigma_i$ is the final request. 
This definition is motivated by a similar characterization
of ``greediness'' in the context of {\em priority algorithms} as defined by Borodin et al.~\cite{DBLP:journals/algorithmica/BorodinNR03}. 

Naturally, not all greedy-like algorithms are expected to be bijectively optimal. For instance, for the classic 
paging problem, all lazy algorithms are greedy-like, however, as shown in~\cite{ADLO07:paging,AS:bijective}, assuming
locality of reference, only LRU is optimal. Therefore, one needs to chose a ``good'' algorithm in this class of
greedy-like algorithms. Let $G$ denote such a greedy-like algorithm. 
We say that $A$ is $G$-like on $\sigma_i$ if, after serving $\sigma[1,i-1]$, $A$ serves request $\sigma_i$ as $G$ would. 
Note that the $G$-like notion cannot be characterized in general for all online problems. It needs to be defined for specific problems
and specific greedy-like algorithms (e.g., the definition of an ``LRU-like'' algorithm in~\cite{AS:bijective}).
Given sequences over ${\cal I}_n$, Algorithm $A$ is  {\em $G$-like on the suffix $[j,n]$}
if $A$ serves all requests $\sigma_j \ldots \sigma_n$ in a $G$-like manner. 
The following definition formally describes algorithms for which the $G$-like decision can be moved
``one step earlier'' without affecting performance with respect to bijective analysis. 
\begin{definition}
Suppose that $A$ is an online algorithm over sequences in ${\cal I}_n$ such that $A$ is $G$-like on the suffix $[j+1,n]$. 
We say that $A$ is {\em $G$-like extendable on $j$} if there exists  
a bijection $\pi:{\cal I}_n \rightarrow {\cal I}_n$ and 
an online algorithm $B$ with the following properties. 
\vspace{-0.5em}
\begin{packed_item}
\item For every $\sigma \in {\cal I}_n$, $B$ makes the same decisions as $A$ on the first $j-1$ requests of $\sigma$.
\item For every $\sigma \in {\cal I}_n$, $B$ is $G$-like on $\sigma_{j}$.
\item $\pi(\sigma)[1,j]=\sigma[1,j]$ and $B(\pi(\sigma)) \leq A(\sigma)$.
\end{packed_item}
\label{property:locally.optimal}
\end{definition}
\vspace{-0.75em}
Informally, $A$ is $G$-like extendable if it can be transformed to another algorithm $B$
that is ``closer'' to the statement of a $G$-like algorithm and is not inferior to $A$ according to bijective analysis.
We note that Definition~\ref{property:locally.optimal} is motivated by the statement of Lemma 3.4 in~\cite{AS:bijective}; 
in contrast to the latter, it applies not only to paging (and the LRU algorithm) but to all online problems for which 
there is a well-defined $G$ algorithm with the above properties (and in particular, is greedy-like). 
This definition is instrumental in proving the optimality of $G$;
in particular, we obtain the following theorem. The proof follows along the lines of the proof of Lemma 3.7 and Theorem 3.8 in~\cite{AS:bijective}.

\begin{theorem}
If every online algorithm $A$ (over requests in ${\cal I}_n$) that is $G$-like on the suffix $[j+1,n]$
is also $G$-like extendable on $j$, for all $1\leq j\leq n$, then $G$ is optimal.
\label{thm:optimality}
\end{theorem}

\begin{proof}
  Consider an arbitrary algorithm $\alg$ and any request sequence $\sigma \in \In$. We will show that there exists a bijection $\pi: \In \to \In$ such that $G(\sigma) \le \alg(\pi(\sigma))$.

Fix an arbitrary $\sigma \in \In$, we will show by reverse induction on the requests that the theorem holds. More precisely, let $\Cls_i$ be the set of all algorithms that are $G$-like on $\sigma[i,n]$ and serve $\sigma[1,i-1]$ exactly as $\alg$. By reverse induction on the indexes of $\sigma$, we show that, for every algorithm $\alg$, there exists an algorithm $C_i \in \Cls_i$ and a bijection $\mu_i$ such that $C_i(\sigma) \le \alg(\mu_i(\sigma))$ and $\mu_i(\sigma)[1,i] = \sigma[1,i]$.

For the last request, define $\mu_n$ to be the identity function, and define $C_n$ to serve $\sigma[1,n-1]$ exactly as $\alg$ and to serve $\sigma_n$ in a $G$-like manner. The claim follows immediately from the fact that $G$ is greedy-like.

Consider the inductive step from $i+1$ to $i$. From the induction hypothesis, there exists an algorithm $C_{i+1} \in \Cls_{i+1}$ such that 
\begin{equation}\label{eq:Cind}
  C_{i+1}(\sigma) \le \alg(\mu_{i+1}(\sigma)) ~.
\end{equation}

By the theorem statement, $C_{i+1}$ is G-like extendable on $i$ and, by Definition~\ref{property:locally.optimal}, we have an algorithm $B_i$ and a bijection $\pi_i$ such that  
\begin{equation}\label{eq:Bi}
  B_i(\sigma) \le C_{i+1}(\pi_i^{-1}(\sigma)) \le \alg(\mu_{i+1}(\pi_i^{-1}(\sigma))) ~, 
\end{equation}
where the last inequality follows from \eqref{eq:Cind}. Note that $\pi_i^{-1}(\sigma)[1,i] = \pi_i(\sigma)[1,i] = \sigma[1,i]$

By applying the induction hypothesis on algorithm $B_i$, there exist an algorithm $C_{i} \in \Cls_{i+1}$ and a bijection $\mu'_{i+1}$ such that 
\begin{equation*}
  C_{i}(\sigma) \le B_i(\mu'_{i+1}(\sigma)) \le \alg(\mu_{i+1}(\pi^{-1}_i(\mu'_{i+1}(\sigma)))) ~,
\end{equation*}
where the last inequality follows from \eqref{eq:Bi}. 

Since $C_i \in \Cls_{i+1}$, $C_i$ is $G$-like on $\sigma[i+1,n]$. Moreover, as $C_i \in \Cls_{i+1}$ and is based on $B_i$, $C_i$ is $G$-like on $\sigma[i]$ as it serves $\sigma[1,i]$ exactly as $B_i$ and, by Definition~\ref{property:locally.optimal}, $B_i$ is $G$-like on $\sigma[i]$. It follows then that $C_i$ serves $\sigma[1,i-1]$ exactly as $\alg$ and is $G$-like on $\sigma[i,n]$. Hence, $C_i \in \Cls_i$. Define $\mu_i := \mu_{i+1} \circ \pi^{-1}_1 \circ \mu'_{i+1}$. Note that $\mu_i(\sigma)[1,i] = \mu_{i+1}(\pi^{-1}_i(\mu'_{i+1}(\sigma[1,i]))) = \sigma[1,i]$ and the inductive step follows. 

After the induction, there exists an algorithm $C_1 \in \Cls_1$. Algorithm $C_1$ is $G$-like on $\sigma[1,n]$ and, therefore, 
$G(\sigma) = C_1(\sigma) \le \alg(\pi(\sigma))$, where $\pi := \mu_1$. 
\end{proof}

We will demonstrate the applicability of this framework by showing optimality of greedy-like online algorithms for three 
well-known online problems: the $2$-server problem on the circle, the weighted paging problem and the reordering buffer management problem.

\subsection{The 2-server problem on the continuous circle}
We begin with the 2-server problem on the continuous circle. Here the candidate algorithm $G$ is the 
obvious greedy algorithm (with an arbitrary tie-breaking rule) that serves a request by moving the server closer to the request, and the $G$-like notion is obviously well-defined. 
\begin{theorem}
\gdy \ is optimal for 2-server on the circle.
\label{thm:circle.2-server}
\end{theorem}

\begin{proof}
Let $A$ denote any online algorithm that is $G$-like on the suffix $[j+1,n]$, for some $j\in[1,n]$. From Theorem~\ref{thm:optimality},
it suffices to prove that $A$ is $G$-like extendable on $j$. We will show the existence of an appropriate online algorithm $B$ and 
a bijection $\pi$, according to Definition~\ref{property:locally.optimal}. In particular, since the definition requires that 
$\pi(\sigma)[1,j]=\sigma[1,j]$, and that $B$ makes the same decisions as $A$ on $\sigma[1,j-1]$, we only need to define $\pi(\sigma)[j+1,n]$,
as well as the decisions of $B$ while serving the latter sequence of requests. 

Consider the request $\sigma_{j}$: if $A$ serves this request in a $G$-like manner (i.e., greedily), 
then the lemma holds trivially. Otherwise, note that after serving 
$\sigma[1,j-1]$ and $\pi(\sigma)[1,j-1]$, respectively, $A$ and $B$ have the same configuration. Namely, if $a_1,a_2$ and $b_1,b_2$ 
denote the servers for the two algorithms at this configuration, we have that $a_1 \equiv b_1$
and $a_2 \equiv b_2$. Since $A$ does not serve $\sigma_j$  greedily, we can assume, without loss of generality, that 
$d(a_2,\sigma_{j}) \geq d(a_1,\sigma_{j})$ and that~$A$ serves the request using~$a_2$ (see Figure~\ref{fig:circle.optimality} for an illustration). 
Let $D=d(a_2,\sigma_{j})-d(a_1,\sigma_{j})$, and let $\overline{\sigma}[j+1,n]$ denote the sequence which is derived from $\sigma[j+1,n]$
by {\em shifting} each request by $d(a_1,\sigma_j)$ in the direction opposite to the move of $a_2$ (in the example of Figure~\ref{fig:circle.optimality}, 
this is done clockwise). We then define the mapping $\pi(\sigma)$ as
$\pi(\sigma)=\sigma[1,j]\cdot \overline{\sigma}[j+1,n]$; it is straightforward to show that this mapping is bijective in ${\cal I}_n$.

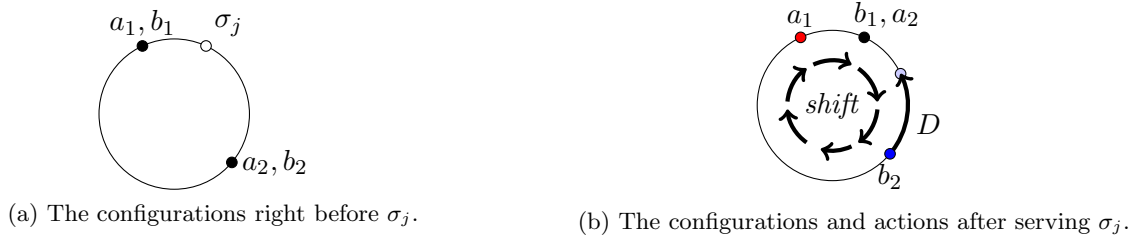
\begin{figure}[htb!]
   \subfloat[The configurations right before $\sigma_j$.]
   {
   \begin{minipage}{0.4\textwidth}
    \centering
    \begin{tikzpicture}[scale=1]
      \draw(0,0) circle (1cm);
      \draw[fill=black] (115:1cm) node[above](a1) {$a_1,b_1$} circle (2pt);
      \draw[fill=black] (320:1cm) node[right](a2) {$a_2,b_2$} circle (2pt);
      \draw[fill=white] (65:1cm) node[above,xshift=0.3cm](ri) {$\sigma_j$} circle (2pt);
    \end{tikzpicture}
    \end{minipage}
    \label{fig:beforeServCircle}
    } \qquad
  \subfloat[The configurations and actions after serving $\sigma_j$.]
  {
    \begin{minipage}{0.5\textwidth}
    \centering
    \begin{tikzpicture}[scale=1]
      \draw(0,0) circle (1cm);
      \draw[fill=red] (115:1cm) node[above](a1) {$a_1$} circle (2pt);
      \draw[fill=black] (65:1cm) node[above,xshift=0.3cm](a2) {$b_1,a_2$} circle (2pt);
      \draw[fill=blue!20] (25:1cm) circle (2pt);
      \draw[<-,ultra thick] (25:1cm) arc[radius=1cm, start angle = 25, end angle = -40] node {}; \node at (1.28,-0.25) {$D$};
      \draw[fill=blue] (320:1cm) node[below](g2) {$b_2$} circle (2pt);
      \foreach \a/\b in {55/5,115/65,175/125,235/185,295/255,355/305}
           \draw[->,ultra thick] (\a:0.6cm) arc[radius=0.6cm, start angle = \a, end angle = \b]; 
       \node at (0,0) {\emph{shift}};
    \end{tikzpicture}
    \end{minipage}
    \label{fig:afterServCircle} 
  }  
  \caption{An illustration of the bijection that \emph{shifts} the requests around the circle by a 
  distance of $d(a_1,\sigma_j)$ and the first action of $A$ after the configurations of $A$ and $B$ diverge on request $\sigma_j$}
  \label{fig:circle.optimality}
\end{figure}

Next, we define the actions of algorithm $B$ over the sequence $\pi(\sigma)[j+1,n]$. In particular, note that 
$B$ serves the request $\sigma_{j}=\pi(\sigma_j)$ greedily; moreover we require that $B$ subsequently moves the server
$b_2$ by a distance equal to $D$ (in the example of Figure~\ref{fig:circle.optimality}, this is done counter-clockwise). 
It can be shown 
by induction on $l$, that, for all $l\in[j+1,n]$, if servers $a_1,a_2$ are at distance $x$ right before $A$ serves $\sigma_l$,
then servers $b_1,b_2$ are at distance $x$ before $B$ serves $\pi(\sigma_l)$; that is, $B$ can serve the request $\pi(\sigma_l)$ by moving one of its
servers that is in the same position, relative to the shift, as the server of $A$ that serves $\sigma_l$, and thus the two costs are identical.  
In conclusion, the cost of $A$ on $\sigma[j+1,n]$ is the same as the cost of $B$ on $\pi(\sigma[j+1,n])$, which further implies that 
$A(\sigma)=B(\pi(\sigma))$, which concludes the proof. 
\end{proof}

\subsection{The weighted paging problem}

Next, we consider the {\em weighted paging problem}. This is a generalization of the standard (uniform) paging problem, in which each page 
$p$ is associated with an eviction cost $c_p$, and has
generated an impressive body of work 
from the point of view of competitive analysis (see e.g.,~\cite{DBLP:journals/siamdm/ChrobakKPV91, DBLP:conf/soda/Young91, 
DBLP:journals/jacm/BansalBN12} and references therein).
It is well known that the weighted paging problem for a cache of size
$k$ is equivalent to the $k$-server problem in a {\em discrete star graph}, assuming there are no requests to the center node of the star. 
More precisely, the star has as many edges as pages, and the weight of each edge
is equal to half the eviction cost of the corresponding page; last, requests may appear on any leaf of the star. 

Consider the simple greedy algorithm $G$ that, 
upon a fault, evicts from the cache a page of smallest cost; clearly, this algorithm is greedy-like. 
The proof of the next theorem relies on Theorem~\ref{thm:optimality} (as in the case of the proof of Theorem~\ref{thm:circle.2-server}).
However, unlike Theorem~\ref{thm:circle.2-server} (and unlike the proof of the bijective optimality of 
greedy/lazy algorithms for unweighted paging in~\cite{ADLO07:paging}), the proof is technically more involved, due to the asymmetry of the 
cost requests (which complicates the argument for the $G$-like extendability of all possible online algorithms).
\begin{theorem}
\gdy \ is bijectively optimal for weighted paging.
\label{thm:weighted.paging}
\end{theorem}

We give the proof in the framework of the $k$-server problem on the discrete star (which as explained, is an equivalent formulation of the
weighted paging problem).
Let $A$ denote any online algorithm that is $G$-like on the suffix $[j+1,n]$ for some $j\in[1,n]$. From Theorem~\ref{thm:optimality},
it suffices to prove that $A$ is $G$-like extendable on $j$. We will show the existence of an appropriate online algorithm $B$ and 
a bijection $\pi$, according to Definition~\ref{property:locally.optimal}. In particular, since the definition requires that 
$\pi(\sigma)[1,j]=\sigma[1,j]$, and that $B$ makes the same decisions as $A$ on $\sigma[1,j-1]$, we only need to define $\pi(\sigma)[j+1,n]$,
as well as the decisions of $B$ while serving the latter sequence of requests. 

Consider the request $\sigma_{j}$: if $A$ serves this request greedily, then the lemma holds trivially. Otherwise, note that after serving 
$\sigma[1,j-1]$ and $\pi(\sigma)[1,j-1]$, respectively, $A$ and $B$ have the same configuration.
For concreteness, let  $a_1, \ldots ,a_k$, and $b_1, \ldots ,b_k$ denote the configurations of $A$ and $B$ after serving 
$\sigma[1,j-1]$, and $\pi(\sigma)[1,j-1]$, respectively, with $a_i \equiv b_i$, for all $i\in[1,k]$. Moreover, let 
$v_i$ denote the nodes on which $a_i$ and $b_i$ lie, right after $A$ and $B$ have served the last request of the sequence 
$\sigma[1,j-1] \equiv \pi(\sigma)[1,j-1]$, and 
let $c_i$ denote the cost of the edge to which $v_i$ is incident.
We can assume, without loss of generality, that $A$ serves request $\sigma_{j}$ by moving the server from position $a_1$, whereas $B$ serves 
request $\pi(\sigma_j)\equiv \sigma_j$ by moving server $b_2$ (hence $c_1 \geq c_2$). We emphasize that indices ``1'' and ``2'' 
will be used throughout the proof to identify
these specific servers $(a_1,b_1,a_2,b_2)$ as well as the nodes $v_1,v_2$, defined as above. 

Given a request $r$ to some node of the star (other than the center), we define $\overline{r}$ as follows:
\[
\overline{r}=
\begin{cases}
a_1, &\quad \text{if} \ r=a_2 \\
a_2, &\quad \text{if} \ r=a_1 \\
r, & \quad \text{otherwise}.
\end{cases}
\]
As a next step, we need to define an appropriate $\pi(\sigma)$ as well as the actions of the algorithm $B$ (relative to the decisions on $A$ on $\sigma$).
We already stipulated that $\pi(\sigma)[1,j]=\sigma[1,j]$, and we have also determined the decisions of $B$ on the first $j$ requests in 
$\pi(\sigma)[1,j]=\sigma[1,j]$. We will next define, in an inductive manner,  both the bijection, 
as well as the decisions of $B$, for all sequences in $\pi(\sigma)[j+1,n]$. For all $\sigma_l \geq j+1$, define $\pi(\sigma_l)$ as follows:
\[
\pi(\sigma_l)=
\begin{cases}
\sigma_l, &\text{$A$ and $B$ have the same configuration after serving $\sigma[1,l-1]$}\\&  \text{and $\pi(\sigma)[1,l-1]$, respectively}, \\
\overline{\sigma}_l, & \text{otherwise}.
\end{cases}
\]
The mapping $\pi(\sigma)$ is then defined, in the natural way, as $\pi(\sigma_1) \ldots \pi(\sigma_n)$. It is straightforward to verify that
this mapping is indeed bijective.

We will next inductively ($l\geq j+1$) define how algorithm $B$ serves request $\pi(\sigma_l)$. 
We first introduce some useful notation. Suppose that after serving sequences $\sigma[1,l]$ and $\pi(\sigma)[1,l]$, respectively, 
$A$ and $B$ have servers at the same node $x$. If $a_q$ denotes the server of $A$ that is located on $x$, then we define $b_{q'}$ to be
$B$'s server that is located on $x$.

We distinguish the following cases, in order to properly define $B$.
\bigskip
\begin{itemize}
\item If $A$ and $B$ have identical configurations right before serving $\sigma_l$ and $\pi(\sigma_l)$, respectively, then both $A$ and $B$
serve their requests identically (and thus are in the same configurations right after serving the corresponding requests).
\item If $A$ and $B$ do not have identical configurations right before serving $\sigma_l$ and $\pi(\sigma_l)$, then we consider the following subcases:
\begin{itemize}
\item {\bf Case 1:} If $\sigma_l\in\{v_1,v_2\}$ and server $a_2$ is at the node of request $\overline{\sigma}_l=\pi(\sigma_l)$, then:
\begin{itemize}
\item {\bf subcase 1a:} If $A$ serves $\sigma_l$ by moving server $a_2$, then $B$ serves $\pi(\sigma_l)$ by moving $b_1$.
\item {\bf subcase 1b:} If $A$ serves $\sigma_l$ by moving server $a_q$ with $q\neq 2$, then $B$ serves $\pi(\sigma_l)$ by moving $b_{q'}$.  
See Figure~\ref{fig:subcase1b} for an illustration.
\end{itemize}
\item {\bf Case 2:} If $\sigma_l\in\{v_1,v_2\}$ and $a_2$ is at the node of request $\sigma_l$ then $B$ serves $\pi(\sigma_l)$ by moving $b_1$. 
\item {\bf Case 3:} If $\sigma_l\notin\{v_1,v_2\}$ then we consider the following subcases:
\begin{itemize}
\item {\bf subcase 3a:} If $A$ serves $\sigma_l$ by moving server $a_q$, with $q\neq 2$, then $B$ moves server $b_{q'}$.
\item {\bf subcase 3b:} If $A$ serves $\sigma_l$ by moving server $a_2$, then $B$ serves $\pi(\sigma_l)$ by moving $b_1$.
See Figure~\ref{fig:subcase3b} for an illustration.
\end{itemize}
\end{itemize}
\end{itemize}
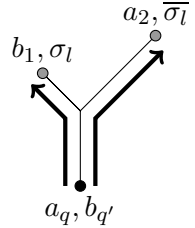
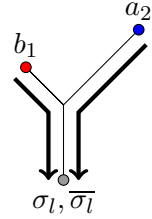
\begin{figure}[htb!]
    \centering
  \subfloat[An illustration of subcase 1b]{ 
     \begin{minipage} {0.4\textwidth}
    \centering
    \begin{tikzpicture}[scale=1]
      \draw (1,1) -- (0,0) -- (-0.5,0.5) -- (0,0) -- (0,-1);
      \draw[fill=black!40] (-0.5,0.5) node[above] {$b_1,\sigma_l$} circle (2pt);
      \draw[fill=black!40] (1,1) node[above] {$a_2,\overline{\sigma_l}$} circle (2pt);
      \draw[fill=black] (0,-1) node[below] {$a_q,b_{q'}$} circle (2pt);
      \begin{scope}[xshift=-0.2cm,yshift=-0.1cm,scale=0.9]
        \draw[->,ultra thick] (0,-1) -- (0,0) -- (-0.5,0.5);
      \end{scope}
      \begin{scope}[xshift=0.2cm,yshift=-0.1cm,scale=0.9]
        \draw[->,ultra thick] (0,-1) -- (0,0) -- (1,1);
      \end{scope}
    \end{tikzpicture}
    \end{minipage}
    \label{fig:subcase1b}
    } 
  \subfloat[An illustration of subcase 3b]{
  \begin{minipage}{0.4\textwidth}
    \centering
    \begin{tikzpicture}[scale=1]
      \draw (1,1) -- (0,0) -- (-0.5,0.5) -- (0,0) -- (0,-1);
      \draw[fill=red] (-0.5,0.5) node[above] {$b_1$} circle (2pt);
      \draw[fill=blue] (1,1) node[above] {$a_2$} circle (2pt);
      \draw[fill=black!40] (0,-1) node[below] {$\sigma_l,\overline{\sigma_l}$} circle (2pt);
      \begin{scope}[xshift=-0.2cm,yshift=-0.1cm,scale=0.9]
        \draw[<-,ultra thick] (0,-1) -- (0,0) -- (-0.5,0.5);
      \end{scope}
      \begin{scope}[xshift=0.2cm,yshift=-0.1cm,scale=0.9]
        \draw[<-,ultra thick] (0,-1) -- (0,0) -- (1,1);
      \end{scope}
    \end{tikzpicture}
    \end{minipage}
    \label{fig:subcase3b}
    }
  \caption{An illustration of subcases 1b and 3a in the statement of algorithm $B$}
  \label{fig:weightedPaging}  
\end{figure}

The following invariant will be instrumental in proving that algorithm $B$ (as shown above) is well-defined, and that $B$ is bijectively no worse 
than $A$.
\begin{lemma} [Invariant]
\begin{enumerate}
\item [(i)] For all requests $\sigma_l$,$\pi(\sigma_l)$, if $A$ and $B$ are not in the same configuration right before serving these requests, respectively,
then either $(a_1,b_2)=(v_1,v_2)$, or $(a_1,b_2)=(v_2,v_1)$. 
\item [(ii)]
Prior to serving $\sigma_l$,$\pi(\sigma_l)$, either $A$ and $B$ (respectively) are in identical configurations, or their configurations
only differ in that $b_1$ is not at a node occupied by a server of $A$, and, likewise, $a_2$ is not at a node occupied by a server of $B$. 
\end{enumerate}
\label{lemma:invariant.paging}
\end{lemma}

\begin{proof}
The proof is by induction on $l$. Suppose that the invariant holds right before $A$ and $B$ serve requests $\sigma_l$ and $\pi(\sigma_l)$, respectively. 
We will show that the invariant holds after these requests are served by verifying that all cases in the statement of $B$ satisfy the 
invariant. For succinctness, we will use the expression ``before/after the requests'' to refer to 
``immediately before/after serving the corresponding requests''.
\begin{itemize}
\item If $A$ and $B$ have identical configurations before the requests, then so they do after the requests, and the invariant holds trivially. 
\item If $A$ and $B$ do not have identical configurations before the requests, we consider the corresponding cases and subcases of algorithm $B$. 
\begin{itemize}
\item {\bf subcase 1a}. This subcase maintains the invariant because if, say $(a_1,b_2)\in (v_1,v_2)$ prior to the requests, then $(a_1,b_2) \in (v_2,v_1)$
after the requests. Similarly if $(a_1,b_2)\in (v_2,v_1)$ before the requests, then  $(a_1,b_2) \in (v_1,v_2)$ after the requests. 
\item{\bf subcase 1b:} After the requests, $A$ and $B$ are in the same configuration, so the invariant is maintained 
(see also Figure~\ref{fig:weightedPaging}).
\item {\bf Case 2:} This case trivially maintains the invariant since $A$ and $B$ do not move any servers. 
\item {\bf subcase 3a:} Part (i) of the invariant holds trivially. Invariant (ii) is maintained because $A$ and $B$ move servers 
from the same node (say $x$) to the same node (say $y$), with $x,y \notin\{v_1,v_2\}.$
\item {\bf subcase 3b:} After the requests, $A$ and $B$ are in the same configuration, thus the invariant is maintained. 
\end{itemize}
\end{itemize}
\end{proof}

We can now use Lemma~\ref{lemma:invariant.paging} (and, in particular, Part (ii)) in order to show that $B$ is well-defined.
More precisely, Part (ii) of the lemma implies that $B$ makes well-defined decisions in case 2, as well as subcases 1b, 3a and 3b; all other 
cases or subcases are trivially well-defined.  

Having established the consistency of $B$, we proceed to the last step of the proof, namely to show that $A(\sigma)\leq B(\pi(\sigma))$.
Let $\delta$ be equal to $A(\sigma_{j})-B(\pi(\sigma_{j}))$; in words, $\delta$ is the difference in the cost incurred by $A$ and $B$ when serving 
$\sigma_{j}$ and $\pi(\sigma_{j})$, respectively (recalling the notation we introduced early in this proof, this cost is equal to $\delta=c_1-c_2$).
The following lemma shows, informally, that as long as $A$ and $B$ are in different configurations, $A$ has payed at least $\delta$ more than $B$,
and when $A$ and $B$ reach the same configuration, $A$ has paid at least as much as $B$.  
\begin{lemma}
Let $\sigma_l$ and $\pi(\sigma_l)$  denote the current requests that are about to be served by $A$ and $B$, respectively. Then 
\begin{itemize}
\item [(i)] If $A$ and $B$ are in the same configuration prior to serving $\sigma_l$ and $\pi(\sigma_l)$, respectively, 
then $A(\sigma[1,l])\geq B(\pi(\sigma[1,l]))$.
\item [(ii)]  If $A$ and $B$ are not at the same configuration prior to serving $\sigma_l$ and $\pi(\sigma_l)$, respectively, then 
$A(\sigma[1,l])\geq B(\pi(\sigma[1,l]))+\delta$.
\end{itemize}
\label{lemma:cost.paging}
\end{lemma}
\begin{proof}
The proof is by induction on $l$. 
Suffices to show that each case in the statement of $B$ maintains the statements (i) and (ii) of the lemma 
after $A$ and $B$ serve requests $\sigma_l$ and $\pi(\sigma_l)$, respectively. 

If $A$ and $B$ are in the same configuration prior to serving $\sigma_l$ and $\pi(\sigma_l)$, then $A$ and $B$ serve the requests identically,
and thus pay the same cost. Hence the lemma is satisfied trivially. 
Otherwise, we consider the remaining cases in the statement of $B$:
\begin{itemize}
\item {\bf subcase 1a:} In this case, the servers of $A$ and $B$ move the same distance, thus the lemma holds. 
\item {\bf subcase 1b:} In this case, $A$ pays on $\sigma_l$, in the worst case, a cost $\delta$ less than 
$B$ on $\pi(\sigma_l)$ (see also Figure~\ref{fig:subcase1b}). Thus the lemma holds.
\item {\bf Case 2:} The lemma holds trivially as both $A$ and $B$ pay zero cost (they have each a server at the corresponding request).
\item {\bf subcase 3a:} $A$ and $B$ serve their requests at the same cost, thus the lemma holds.
\item {\bf subcase 3b:} This case is similar to subcase 1b: namely, $A$ pays on $\sigma_l$, in the worst case, a cost $\delta$ less than $B$ 
pays on $\pi(\sigma_l)$. 
\end{itemize}
\end{proof}
The following corollary completes the proof that $A$ is $G$-like extendable and shows that the greedy algorithm
is optimal. This completes the proof of Theorem~\ref{thm:weighted.paging}.

\begin{corollary}
$A(\sigma)\geq B(\pi(\sigma))$.
\label{cor:cost.paging}
\end{corollary}
\begin{proof}
Suppose there is an index $l$ such that, after serving $\sigma[1,l]$ and $\pi(\sigma[1,l])$, $A$ and $B$ are in the same configuration. Then 
from Lemma~\ref{lemma:cost.paging}, $A(\sigma[1,l])\geq B(\pi(\sigma[1,l]))$. For all subsequent requests in $\sigma[l+1,n]$ and $\pi(\sigma)[l+1,n]$,
from the statement of $B$, we deduce that $A$ and $B$ remain in the same configuration; furthermore, $A$ serves $\sigma(h)$ in the same way
as $B$ serves $\pi(\sigma(h))$ (for all $l \leq h\leq n$), and thus $A(\sigma(h))=B(\pi(\sigma(h))$. 
Otherwise, there is no such index $l$ such that, after serving $\sigma[1,l]$ and $\pi(\sigma[1,l])$, $A$ and $B$ are in the same configuration, and Lemma~\ref{lemma:cost.paging} shows that $A(\sigma)\geq B(\pi(\sigma))+\delta> B(\pi(\sigma))$.  
In both cases, we obtain that $A(\sigma)\geq B(\pi(\sigma))$.
\end{proof}

\subsection{Reordering buffer management}

As a third application of our framework, we consider the well-studied {\em reordering buffer management} problem, introduced by 
R\"acke et al.~\cite{DBLP:conf/esa/RackeSW02}. 
It consists of a service station that has some active colour, an initially empty buffer of size $k$ and a sequence of coloured requests. Requests enter the buffer sequentially and all items within the buffer that are the same colour as the active colour can be served  
by the service station. If none of the items in the buffer have the active colour, the service station must change its active colour at a fixed
cost. The goal is to minimize the number of colour switches.
As with the other problems we consider in this
paper, the reordering buffer management problem has been studied extensively in the context of competitive analysis 
(see, e.g.~\cite{DBLP:conf/stoc/AdamaszekCER11,DBLP:journals/talg/Avigdor-Elgrabli15,DBLP:conf/focs/Avigdor-ElgrabliR13} and references therein).

For this problem, we define $G$ as the greedy algorithm that switches (only if necessary) to a colour $c$ 
for which the number of items of colour $c$ in the buffer
is maximized among all colours (and is thus trivially greedy-like). 
We once again rely on Theorem~\ref{thm:optimality} in order to show bijective optimality. 
The nature of this problem gives rise to some technical complications in the optimality
proof, in the sense that an algorithm may delay processing a request,
(an option that is not meaningful in the context of paging/$k$-server problems),
which in turn complicates the comparison of $A(\sigma)$ and $B(\pi(\sigma))$ on a request-by-request manner. 

\begin{theorem}
\gdy \  is bijectively optimal for reordering buffer management.
\label{thm:buffer.optimality}
\end{theorem}

For the reordering buffer management problem, the next item in the request sequence to enter the buffer is called the \emph{current request}. It is useful to define some notion of time for this problem. At time step $i$, the current request is $\sigma_{i+1}$. That is, the items $\sigma_1,\ldots,\sigma_{i}$ have entered the buffered (and possible have been served). More precisely, ``at time step $i$'' refers to the precise moment that $\sigma_i$ enters the buffer. Moreover, without loss of generality, we assume that the current request enters the buffer as soon as a slot is freed.
At time $n$, all the items have entered the buffer and there is no current request.

For this problem, there is a natural notion of ``laziness'', as defined in \cite{DBLP:conf/esa/RackeSW02}, where the algorithm only changes its active colour when otherwise it can no longer make any progress in the input. Without loss of generality, we can assume that all the algorithms are lazy \cite{DBLP:conf/esa/RackeSW02}.

With these notions, we get the following observation that will be useful for Theorem~\ref{thm:buffer.optimality}.

\begin{observation}\label{obs:canServNotIn}
  For some $\tau < \tau'$, consider any lazy algorithm $A$ at time $\tau$ and any lazy algorithm $A'$ at time $\tau'$ for a given request sequence $\sigma$. All the items in the buffer of $A$ either have been served by $A'$ or are in the buffer of $A'$. This implies that, if $A$ switches to a colour $c$ that is not in the buffer of $A'$ at time $\tau'$, $A$ advances to some time $\tau'' \le \tau'$.
\end{observation}

As in \cite{DBLP:conf/esa/RackeSW02}, a \emph{colour block} is the set of items of the same colour that are served with a single colour switch whereas a \emph{buffer colour block} is the set of items of the same colour in the buffer. For some algorithm $A$, $\mathcal{Z}^A_\tau(\sigma)$ are the buffer colour blocks at time $\tau$ for $\sigma$.

Let $A$ denote any online algorithm that is $G$-like on the suffix $\sigma[j+1,n]$ for some $j \in [1,n]$. 
From Theorem~\ref{thm:optimality}, it suffices to prove that $A$ is $G$-like extendable on $j$. We will show the existence of an appropriate online algorithm $B$ and a bijection $\pi$, according to Definition~\ref{property:locally.optimal}. In particular, since the definition requires that $\pi(\sigma)[1,j] = \sigma[1,j]$, and since $B$ makes the same decisions as $A$ on $\sigma[1,j-1]$, we only need to define $\pi(\sigma)[j+1,n]$, as well as the decisions of $B$ while serving the latter sequence of requests. If $A$ is $G$-like on $\sigma[i+1]$, we can define $B$ as $A$ and the claim follows. Hence, for the rest of the proof, we assume that $A$ is not $G$-like on $\sigma_i$.

Let $x$ be the active colour of $A$ after the colour switch at time $i$ and let $y \ne x$ be a colour of maximum cardinality in the buffer of $A$ at time $i$. 

In order to define $\pi$, we define a bijective mapping. Given a colour (or a request) $r$, we define $\overline{r}$ as follows.
$$
\overline{r} = \begin{cases}
          x, &\text{if } r = y \\ 
          y, &\text{if } r = x \\
          r, &\text{otherwise.}
          \end{cases}
$$
Given a sequence of requests $\sigma = \left<\sigma_1,\ldots,\sigma_n\right>$, we define $\overline{\sigma} = \left<\overline{\sigma_1},\ldots,\overline{\sigma_n}\right>$.
The bijection $\pi$ is defined such that all the request after $j$ have the colours $x$ changed to $y$ and $y$ changed to $x$. More formally,
$\pi(\sigma) = \sigma[1,j]\cdot\overline{\sigma}[j+1,n]$.

Now, we will define the actions of $B$ for $\pi(\sigma)$. For the requests $\pi(\sigma)[1,j] = \sigma[1,j]$, Algorithm $B$ performs the same colour changes as $A$. Note that this ensures that the contents of the buffers of $A$ and $B$ are the same at time $j$. For request $\sigma_{j}$, $B$ must make a colour switch since $A$ makes a colour switch. Algorithm $B$ switches to colour $y$, a colour of maximum cardinality in the buffer. Let $\Phi = |Z_y| - |Z_x|$, where $Z_y, Z_x \in \mathcal{Z}^A_j$ are the buffer blocks of colour $y$ and $x$.

For the remaining requests, $B$ will simulate $A$ on $\sigma$ and maintain a queue of the colour switches of $A$ after time $j$, where a colour switch to colour $c$ by $A$ is enqueued as $\overline{c}$. Whenever $B$ must make a colour switch, it removes colours from the queue until the dequeued colour matches the colour of a request in the buffer. This ensures that $B$ will not make more colour switches than $A$. In the following lemma, we show that $B$ is well-defined and serves all the remaining requests.

\begin{lemma}\label{lem:bufferWellDef}
  Algorithm $B$ is well-defined over $\pi(\sigma)$.
\end{lemma}

\begin{proof}
  Algorithm $B$ is well-defined on $\pi(\sigma)[1,j] = \sigma(1,j)$ as it performs the same actions as $A$ does over $\sigma(1,j-1)$ and then switches to a colour in the buffer at $\sigma_j$. 

Now, consider the request sequence $\pi(\sigma[j+1,n])$. 
  Let $Q$ be the sequence of colours that $B$ dequeues from its queue
  over $\pi(\sigma)$, i.e., the list of colour switches of $A$ after
  time $j$. We will show, by induction on the indexes of $Q$, that $B$
  is well-defined and will be able to serve the remaining requests of
  $\pi(\sigma)$.  Specifically, we will show that the following
  invariants are maintained throughout.

  For a dequeued colour $c$, let $\tau'$ be the time of $B$ when $c$
  is dequeued and let $\tau$ be the time step of the colour switch in
  $A$. Let $Y^\alg_\tau$ (resp. $X^\alg_\tau$) be the items of colour
  $y$ (resp. $x$) in the buffer at time $\tau$ with an index more than
  $j$.
  The following invariants imply the correctness and completeness of
  $B$:
  \begin{enumerate}
  \item $\tau' \ge \tau$;
  \item for every $Z \in \mathcal{Z}^A_\tau(\sigma)$, if there exists
    a $Z' \in \mathcal{Z}^B_{\tau'}(\pi(\sigma))$ of the same colour,
    then $|Z| \le |Z'|$ and, for every $\sigma_q \in Z$,
    $\pi(\sigma)_q \in Z'$; and
  \item $|Y^A_\tau| \le |X^B_{\tau'}|$ and
    $|X^A_\tau| \le |Y^B_{\tau'}|$ and, for every
    $\sigma_q \in Y^A_\tau \cup X^A_\tau$,
    $\pi(\sigma)_q \in X^B_{\tau'} \cup Y^B_{\tau'}$.
  \end{enumerate}

  Prior to the colour switch at time $i+1$, the buffers of $A$ and $B$
  are identical. After serving $X^A$, Algorithm $A$ is at time $\tau$
  and, after serving $Y^B$, Algorithm $B$ is at time
  $\tau' \ge \tau + \Phi$. Moreover, this means that $B$ will read
  more requests from the input than $A$ and, hence, the second and
  third invariants hold.

  Assume that the invariants hold from index $1,\ldots,\ell-1$ and
  consider the $\ell$-th element of $Q$. There are three cases to
  consider: (1) there are no items of colour $Q[\ell]$ in the buffer
  of $B$ at time $\tau'$, (2) there are items of colour
  $Q[\ell] \ne x$ in the buffer of $B$ at time $\tau'$, and (3) there
  are items of colour $Q[\ell] = x$ in the buffer of $B$ at time
  $\tau'$.

  \begin{itemize}
  \item \textbf{Case 1:} There are no requests of colour $Q[\ell]$ in
    the buffer of $B$ at time $\tau'$.
 \begin{itemize}
    \item \textbf{Invariant 1:} By
      Observation~\ref{obs:canServNotIn}\footnote{Note that, at this
        point, the sequences $\sigma$ and $\pi(\sigma)$ are the same
        modulo a relabelling of the colour $x$ to $y$ and $y$ to
        $x$.}, after the colour switch of $Q[\ell]$, $A$ is at time
      $\tau'' \le \tau'$.

    \item \textbf{Invariant 2:} Assume for contradiction that items of
      a colour $c$ enter the buffer of $A$ during this colour switch
      such that, for $Z_c \in \mathcal{Z}^A_\tau(\sigma)$ and
      $Z'_c \in \mathcal{Z}^B_{\tau'}(\pi(\sigma))$, $|Z_c| > |Z'_c|$.
      This can only occur if there exists an $\sigma_j \in Z_c$ with
      an index greater than all the requests in $Z'_c$, but that would
      contradict the fact that $\tau'' \le \tau'$.

    \item \textbf{Invariant 3:} By a similar argument, after time $i$,
      every request that enters the buffer of $A$ with colour $y$
      (resp. $x$) must be in the buffer of $B$ (but with colour
      $\bar{y} = x$ (resp. $\bar{x} = y$), maintaining Invariant 3.
    \end{itemize}

  \item \textbf{Case 2:} There are requests of colour $Q[\ell] \ne x$
    in the buffer of $B$ at time $\tau'$.  By Invariant 2, for
    $Z_{Q[\ell]} \in \mathcal{Z}^A_\tau(\sigma)$ and
    $Z'_{Q[\ell]} \in \mathcal{Z}^B_{\tau'}(\pi(\sigma))$,
    $|Z'_{Q[\ell]}| \ge |Z_c|$, and $B$ will make at least as much
    progress in the request sequence as $A$. This guarantees all three
    invariants. The same argument holds for $Q[\ell] = y$ from
    Invariant 3.

  \item \textbf{Case 3:} 
    \begin{itemize}
    \item There are requests of colour $Q[\ell] = x$ in the buffer of
      $B$ at time $\tau'$.  If this is the first colour switch to $x$
      after $i$, $A$ will read at most $\Phi$ requests that are not of
      colour $y$ that have already been read by $B$, taking $A$ to
      time $\tau'' \le \tau'$. At $\tau''$, there is additional space
      for $|Y^A_{\tau}|$ items not of colour $y$ in the space of the
      buffer of $A$.  By Invariant 3, $|Y^A_{\tau}|$ is no greater than
      the size of $X^B_{\tau'}$ in the buffer of $B$. The progress
      made by $A$ after time $\tau''$ is no greater than the progress
      made by $B$ after time $\tau'$, guaranteeing all three
      invariants.

    \item If this is not the first colour switch to $x$ after $i$, by
      Invariant 3, $|Y^A_{\tau}|$ is no greater than the size of
      $X^B_{\tau'}$ in the buffer of $B$. Again, the progress made by
      $A$ is no greater than the progress made by $B$, guaranteeing all
      three invariants.
    \end{itemize}
\end{itemize}
\end{proof}

\begin{proof}[Proof of Theorem~\ref{thm:buffer.optimality}]
  From Lemma~\ref{lem:bufferWellDef}, the algorithm $B$ is well-defined. By the definition of $B$, the number of colour switches of $B$ over $\pi(\sigma)[1,j]$ is the same as of $A$ over $\sigma[1,j]$ and the number of colours switches of $B$ over $\pi(\sigma)[j+1,n]$ is no more than $A$ over $\sigma[j+1,n]$. Hence, $B(\pi(\sigma)) \le A(\sigma)$.
\end{proof}

\section{The bijective ratio of the $k$-server problem on the line and the circle}
\label{sec:lineandcircle}
\subsection{Lower bounds}
\label{subsec:lower.bounds}
While Theorem~\ref{thm:circle.2-server} shows that \gdy \ is bijectively optimal for 2-server on the circle, a similar statement does not hold for
the case of the line metric. In fact, in Theorem~\ref{thm:greedy.nonoptimal} we prove a stronger statement,
namely, we design an explicit online algorithm $A$ which has a lower average-cost ratio against \gdy \
of at most $c$ for some constant $c<1$.
\begin{theorem}
For 2-server on the line, and requests over 
${\cal I}_n$, there is an online algorithm $A$ and constants $c,c'$ with $c<1$, 
such that $\sum_{\sigma \in {\cal I}_n} A(\sigma) \leq c\sum_{\sigma \in {\cal I}_n} \gdy(\sigma)+c'$.
\label{thm:greedy.nonoptimal}
\end{theorem}

\begin{proof}
Suffices to show that, if we chose a sequence $\sigma \in {\cal I}_n$ uniformly at random, then 
$\mathbb{E}(A(\sigma)) \leq c\cdot \mathbb{E}(\gdy(\sigma))+c'$.
Let~$x_1, x_2\in [0,1]$ with~$x_1<x_2$ be the two server positions of a configuration~$C$.
We say that~$C$ is~\emph{unfavourable} for \gdy \ if~$x_1\leq t/3$ and~$(2/3)t \leq x_2\leq t$ with~$t\in (0,1)$ a small constant that we will choose later.

We prove the theorem in two steps. The first step shows that there is an algorithm that essentially simulates \gdy, but, when faced with an unfavourable situation for \gdy, it can outperform \gdy \ by a constant factor. The second step is then to show that after starting in an arbitrary configuration, 
\gdy \ will find itself in an unfavourable configuration within a constant number of requests with some probability~$p>0$.

We now formally define Step 1 and Step 2.

\smallskip
\noindent
\emph{Step 1.} Here we show that there is a~$c_1<1$ for which the following holds. If~$C$ is a configuration that is unfavourable for $\gdy$ then there is an algorithm~$A$ such that for request sequences of length exactly~$3$ we have~$\mathbb{E}(A) \leq c_1 \cdot \mathbb{E}(\gdy)$ and the final configuration after the three requests are processed by~$A$ or by~$\gdy$ is the same.

\smallskip
\noindent
\emph{Step 2.} Here we show that, for every fixed~$t\in (0,1)$, starting from an arbitrary configuration~$(y_1,y_2)$, $\gdy$ will find itself in an unfavourable configuration~$(x_1,x_2)$ with~$x_2<t$ within a constant number of steps (depending on~$t$) with positive probability~$p>0$ (also depending on~$t$).

\medskip
We now proceed with the details in the analysis of the two steps. 

\smallskip
\noindent
{\em Analysis of Step 1:}
Algorithm~$A$ works as follows. Let~$\sigma_1,\sigma_2,\sigma_3$ be three requests. Define the positions
\begin{align} 
  x_3 &= (5/8) x_2 + (3/8) x_1 \text{, and} \label{eq:x3} \\
  x_4 &= 10 x_2 \label{eq:x4}
\end{align}
(see Figure~\ref{outperform:greedy:on:average}).
On the first request, Algorithm~$A$ uses Server~1 (at that time positioned at~$x_1$) to serve~$\sigma_1$ if~$(1/2) x_1 + (1/2) x_2 \leq \sigma_1 \leq x_3$.  In this case we say~$A$ was successful (at outperforming $\gdy$) for the first request. If~$\sigma_1$ does not satisfy these inequalities then Algorithm~$A$ simply simulates $\gdy$ on all requests. 

\begin{figure}[htb!]
\caption{The figure depicts an unfavourable configuration for $\gdy$ and the intervals within which the three subsequent requests need to appear so that the new algorithm outperforms $\gdy$.}
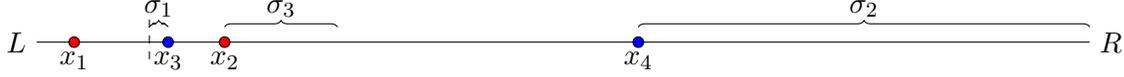
\label{outperform:greedy:on:average}
\begin{tikzpicture}[scale = 1]
   \draw (0,\lG) node[left] {$L$} --(14,\lG) node[right] {$R$};
  \draw[fill=red] (0.5,\lG) node[below](u1) {$x_1$} circle (2pt);
  \draw[fill=red] (2.5,\lG) node[below](u2) {$x_2$} circle (2pt);
   \draw[fill=blue] (1.75,\lG) node[below](u1) {$x_3$} circle (2pt); 
\draw [dashed] (1.5,\lG+0.3) -- (1.5,\lG-0.3);
   \draw[fill=blue] (8,\lG) node[below](u1) {$x_4$} circle (2pt); 
   
  \draw[decorate,decoration={brace}] ($(1.5,\lG) + (0,0.2)$) to node[midway,above,align=center] {$\sigma_1$} ($(1.75,\lG) + (0,0.2)$);
\draw[decorate,decoration={brace}] ($(8,\lG) + (0,0.2)$) to node[midway,above,align=center] {$\sigma_2$} ($(14,\lG) + (0,0.2)$);
  \draw[decorate,decoration={brace}] ($(2.5,\lG) + (0,0.2)$) to node[midway,above,align=center] {$\sigma_3$} ($(4,\lG) + (0,0.2)$);
\end{tikzpicture}
\end{figure}

On the second request, if~$A$ was successful for~$\sigma_1$, then the algorithm serves request~$\sigma_2$ using Server 2 whenever~$x_4 \leq\sigma_2$. We say that~$A$ was successful on the second request. In every other case,  Algorithm~$A$ simulates $\gdy$. For that purpose, it places both servers to the positions they would be in if $\gdy$ had been executed on~$\sigma_1 \sigma_2$ to begin with. 

On the third request, Algorithm~$A$ simply simulates $\gdy$, again meaning that it places both servers to the positions in which they would be if $\gdy$ had been executed on~$\sigma_1 \sigma_2 \sigma_3$. 

We will now compare the costs of~$A$ and the $\gdy$ algorithm.
If~$A$ is unsuccessful in the first step, then~$A$ and $\gdy$ have the same costs. The probability that~$A$ is successful in the first step is at least~$x_3 -(x_2/2+x_1/2) > k$ for some constant~$k>0$. 
Assuming thus that~$A$ is successful on the first request,~$A$ incurs a cost that is at most~$d(x_3,x_1) - d(x_2,x_1)$ larger than the cost of $\gdy$. Denoting by~$D_i$ the average of the cost of~$A$ minus the cost of $\gdy$ in step~$i\in\{1,2,3\}$, we obtain
\[D_1 \leq k\cdot \big( d(x_3,x_1) - d(x_2,x_1)\big).\]
We compare the cost of the second step under the assumption that~$A$ was successful in the first step. If~$A$ is successful in the second step, then~$A$ incurs a cost that is at least $d(x_2,x_3)$ smaller than the one of $\gdy$. This happens with probability~$(1-x_4)$. If~$A$ is unsuccessful, then~$A$ has a cost that is at most~$d(x_2,x_1)$ larger than the one of $\gdy$. We obtain an average difference in cost of
\[D_2 \leq k \cdot  \big(x_4 \cdot d(x_2,x_3) - (1-x_4)\cdot  d(x_2,x_3) \big).\]
We now compare the costs for the third request. If~$A$ was unsuccessful in one of the previous steps, then the cost of~$A$ on the third request is the same as that of $\gdy$. Otherwise, we consider two cases. If~$x_2\leq\sigma_3\leq (\sigma_2-x_1) /2$, which happens with probability at least~$(x_4-x_2)/2$, then $\gdy$ serves~$\sigma_3$ by moving Server 1. In comparison to the cost of $\gdy$, algorithm~$A$ saves at least~$d(x_1,x_2)/2$.
Otherwise, if~$\sigma_3$ is outside of said range, the cost of~$A$ is at most~$d(x_3,x_1)$ greater than the cost of $\gdy$ because the total difference of the two configurations is at most~$d(x_3,x_1)$. This happens with probability at most~$x_4$. We obtain
\[D_3 \leq k \cdot  (1-x_4) \cdot \big( x_4 d(x_3,x_1) - (x_4-x_2)/2 \cdot  d(x_1,x_2)/2 \big).\]
Substituting~$x_3$ and~$x_4$, using \eqref{eq:x3} and \eqref{eq:x4} respectively, overall we obtain that~$D_1+D_2+D_3 \leq -(1/2) k (x_2-x_1) x_2 (23- 80 x_2)$.
Thus, if~$t$ (and hence $x_2$) is sufficiently small, then, on average, the cost of algorithm~$A$ is smaller by a constant amount than the cost of $\gdy$.

\bigskip
\noindent
\emph{Analysis of Step 2:} 
Let~$t$ be a fixed number in~$(0,1)$. 
Starting from an arbitrary configuration~$C$ with constant probability after the first request, server two is in the right half of the line. Assuming this, after the second request, with probability at least~$t/3$, Server 1 is located at a position~$x_1\leq t/3$. On subsequent requests as long as Server~$2$ is located to the right of~$t$, there is a positive probability that Server 2 moves a distance of at least~$(1/6) t$ towards Server~$1$ but not beyond~$(2/3)t$. 
 Thus, with positive probability after at most~$6/t$ steps Server 2 is at a position with~$(2/3)t\leq x_2\leq t$ and Server 1 has not moved.
\end{proof}

Next, we will show a lower bound on the bijective ratio of {\em any} lazy, deterministic online algorithm for the 2-server problem on the line
(which also extends to general $k$-server on both the circle and the line). 
The following proposition is useful in establishing lower bounds on the bijective ratio of a given online algorithm $A$
against an algorithm $B$ that may be online or offline.
Its proof follows from the fact that under any bijection $\pi$, there is at least one sequence $\sigma$
such that $B(\pi(\sigma))\leq c$ and $A(\sigma)\geq \rho c$.
\begin{proposition}
Suppose that there are $c>0$ and $n_0$ such that, for all $n\geq n_0$, 
$ |\{\sigma \in {\cal I}_n : A(\sigma) < \rho \cdot c\}| < |\{\sigma \in {\cal I}_n : B(\sigma) \leq c\}|$. Then the bijective ratio 
of $A$ against $B$ is at least $\rho$.
\label{prop:lower.bound}
\end{proposition}

\begin{theorem}
Any deterministic online algorithm for 2-server even on a line metric of three equidistant points has strict offline bijective ratio at least 2.
\label{thm:2.bound}
\end{theorem}

\begin{proof}
We will show that there exists an initial server configuration which yields the bound.
Suppose that the three points on the line metric are numbered 1,2,3, from left to right, and that the initial 
server positions are at the end points of the line, i.e., points 1 and 3. 
Let also $d$ denote the distance between two consecutive points on the line. For a given $n\geq n_0$, define 
the set of sequences $S \subseteq {\cal I}_n$ such that $\sigma[i]\in \{1,3\}$, for all $i\leq n-2$, 
$\sigma[n-1] \in \{2\}$, and $\sigma[n]\in\{1,3\}$. The optimal offline algorithm can serve every sequence in $S$
at a cost equal to $d$. In contrast, there exists a $\sigma \in S$ such that any 
deterministic online algorithm $A$ must pay at least $2d$ to serve $\sigma$. Namely, if $A$ serves the request $\sigma[n-1]$ 
by moving server 1, the sequence $\sigma$ with $\sigma[n]\in\{1\}$ has this property (symmetrically if $A$ serves the request $\sigma[n-1]$ 
by moving server 3). Last, note that $d$ is the cheapest non-zero cost at which a sequence in ${\cal I}_n$ can be served.
We thus obtain that $ |\{\sigma \in {\cal I}_n : A(\sigma)< 2 \cdot d\}| < |\{\sigma \in {\cal I}_n : \opt(\sigma) \leq d\}|$, and the 
theorem follows from Proposition~\ref{prop:lower.bound}.
\end{proof}
We note that the lower bound of Theorem~\ref{thm:2.bound} extends to general $k$-server on both the circle and the line. 

Last, in the following theorem, we show a lower bound on $\gdy$ under the Max/Max ratio for the line which implies a lower bound on the bijective ratio.
\begin{theorem}
  For any $\eps > 0$, the bijective ratio of $\gdy$ is at least $\frac{k}{2} - \eps$ for the line and at least $\frac{k}{3} - \eps$ for the circle.
\label{lem:linek2}
\end{theorem}

\begin{proof}
We show the result for the line. (The bound for the circle follows with a similar argument.)  
Consider a line that runs from $0$ at the left-most point to $1$ at the right-most point. 
Let the servers be labelled $s_1,s_2,\ldots,s_k$ from left to right (in any configuration, including the initial).
First, we show that it is possible to force $\gdy$ to move all the servers close to $0$ in a constant number of requests. Specifically, 
define $\delta := \eps/k$. For any constant $\delta'$, $\frac{2\delta}{k-1} \ge \delta' > 0$, we force $\gdy$ to move $s_1$ to $0$ and to move all the other servers towards $0$ so that $d(s_i,s_{1+1}) < \delta'$. An initial request is placed at $0$, forcing $\gdy$ to move $s_1$ to $0$. Iterating on $i$ from $2$ to $k-1$, requests are placed as close as possible to the mid-point between $s_{i-1}$ and $s_i$ so that $\gdy$ moves $s_i$, continuing until $d(s_{i-1},s_i) < \delta'$. In total, this requires at most $\lceil\log(1/\delta')\rceil(k-2) + 1$ requests.

The remaining requests alternate between $x := \frac{1}{2} + \frac{k-1}{2}\delta'$ and 1. Note that $x$ is past the mid-point of $1$ and $s_{k-1}$. Hence, these requests (except possibly the first one) will be served by $s_k$ at a cost of more than $1/2 - \delta = 1/2 - \eps/k$. 

The worst-case for \kCtr \ occurs when all the requests are to an extreme point from any server. This has a total cost of $1/k$ ($1/(2k)$ to serve the request and $1/(2k)$ to return). For large enough $n$, this gives a Max/Max ratio of $k/2 - \eps$ which implies the theorem. 

By a similar argument, it is possible to force a server of $\gdy$ to travel at least $1/3 - \eps/k$ on the circle for all but a constant number of requests.
\end{proof}

\subsection{Upper bounds}
\label{subsec:upper.bounds}
We now consider upper bounds and provide sufficient conditions for showing that 
the bijective ratio of an online algorithm $A$ against an algorithm $B$ (that may be online or offline) is at most $c$;
these conditions are formally described in Lemma~\ref{lem:genApproach} and Lemma~\ref{lem:amort}. Both lemmas require 
two conditions stated in terms of individual requests, and their combination yields the desired bound. We use the 
notation $A(\sigma[i] | B(\sigma[1,i-1]))$ to denote the cost of $A$ for serving request $\sigma[i]$ assuming a configuration
resulting from $B$ serving sequence $\sigma[1,i-1]$, where the suffix of the request sequence is implied. For the $k$-server problem in particular, and the algorithms we consider, 
this is a well-defined concept. Since $B$ can be either offline or online, we further emphasize that, in the former case, it
is implicit that the decisions of $B$ on $\sigma[1,j]$ are contingent on its acting on the entire sequence $\sigma[1,n]$.
Last, we use $A(\sigma[i])$ to denote $A(\sigma[i]|A(\sigma[1,i-1]))$. 

\begin{lemma}\label{lem:genApproach}
  Suppose that there exists a $c>1$, a $d > 0$ and a bijection $\pi$ over ${\cal I}_n$ such that, given an online algorithm $A$ and an algorithm 
  $B$, for all $\sigma \in {\cal I}_n$ and all $i \leq n$, the following hold: 
    \vspace{0.5em} \\  
    \indent
 \begin{minipage}{0.9\textwidth}
  \begin{packed_item}
  \item[(i)]$A(\sigma[i] | B(\sigma[1,i-1])) \le d \cdot B(\sigma[i])$, and
  \item[(ii)]$A(\sigma[i]) - A(\pi(\sigma)[i]| B(\pi(\sigma)[1,i-1])) \leq (c-d) \cdot B(\pi(\sigma)[i])$, 
  \end{packed_item}
    \end{minipage} \vspace{0.5em} \\ 
then, $A(\sigma) \le c \cdot B(\pi(\sigma))$.
\end{lemma}

\begin{proof}
  Using $\pi(\sigma)$ as the request sequence for (i) and adding both inequalities, we get $A(\sigma[i]) \le c \cdot B(\pi(\sigma)[i])$. The lemma follows by summing over all the requests.
\end{proof}

The following lemma formalizes a similar approach using amortized analysis.

\begin{lemma}
\label{lem:amort}
Given an online algorithm $A$ and an algorithm $B$,
let $\Phi$ be any potential function such that the amortized cost of $A$ for $\sigma[i]$ is $a_i = A(\sigma[i]) + \Delta\Phi_i$, 
where $\Delta\Phi_i = \Phi_i - \Phi_{i-1}$, and $\Phi_0$ is the potential prior to serving the first request. Suppose also 
that there exist $c,d>0$ and a bijection $\pi$ over ${\cal I}_n$ such that, for all $\sigma \in {\cal I}_n$ and all $i \leq n$, the following hold:
      \vspace{0.5em} \\  \indent
  \begin{minipage}{0.9\textwidth}
   \begin{packed_item}
   \item[(i)]$A(\sigma[i] | B(\sigma[1,i-1])) \le d \cdot B(\sigma[i])$, and 
   \item[(ii)] $a_i \le c \cdot A(\pi(\sigma)[i] | B(\pi(\sigma)[1,i-1]))$, 
   \end{packed_item}
   \end{minipage} \vspace{0.5em} \\ 
  then, $A(\sigma) \le c \cdot d \cdot B(\pi(\sigma)) + \Phi_0 - \Phi_n$.
\end{lemma}
 
\begin{proof}
Summing Inequality (ii) over all the requests gives  
$$\sum_{i=1}^n a_i = A(\sigma) + \Phi_n - \Phi_0 \le c \sum_{i=1}^n A(\pi(\sigma)[i] | B(\pi(\sigma)[1,i-1]))~.$$
Thus,
$$
  A(\sigma) \le c \left(\sum_{i=1}^n A(\pi(\sigma)[i] | B(\pi(\sigma)[1,i-1]))\right) + \Phi_0 - \Phi_n
            \le cd \cdot B(\pi(\sigma)) + \Phi_0 - \Phi_n ~,
$$
where the last inequality follows from (i), using $\pi(\sigma)$ as the request sequence.
\end{proof}

\subsubsection{Defining the bijection}
\label{subsec:ordered.bijection}
This section addresses the definition of a suitable bijection.
Let $\delta > 0$ be a positive value, representing the distance between any two adjacent points in the metric space. 
For a given server configuration $C$, let $\cP_C^\delta$ be the sequence of points in the metric space, ordered by their
distance from the closest server in $C$. 
That is, for all $i \leq |\cP_C^\delta|$, $\minD(\cP_C^\delta[i]) \le \minD(\cP_C^\delta[i+1])$, where $\minD(\cP_C^\delta[i])$ 
is the distance from $\cP_C^\delta[i]$ to the nearest server in configuration $C$. 
\begin{definition}
  For the server configurations $C_1$ and $C_2$, 
  an \emph{ordered bijection} $(\omb)$ is a bijection such that, for all $i$, $\cP_{C_1}^\delta[i]$ is matched to $\cP_{C_2}^\delta[i]$.
  Let $A$ and $B$ be two algorithms; we define $\pi^{A,B}(\sigma)=\pi^{A,B}(\sigma[1] \ldots \sigma[n])$ to be any bijection of the form $\pi^{A,B}_1(\sigma[1]) \ldots \pi^{A,B}_n(\sigma[n])$
  such that $\pi^{A,B}_i$ is the \omb \ of the server configurations of $A$ and $B$ right after serving the sequences $\sigma[1,i]$ and 
  $\pi^{A,B}_1(\sigma[1]) \ldots \pi^{A,B}_{i-1}(\sigma[i-1])$, respectively. Note that this definition is applicable to all metric spaces. (When clear from the context, we drop the superscript of~$\pi$.)
  \label{def:ob}
\end{definition}

As there may be multiple points at some distance $d$ from the nearest server, each permutation of the points at 
distance $d$ represents a different bijection. In \omb, we assume that ties are broken arbitrarily.
We assume for simplicity that both the line and the circle have unit length and that $\delta$ is chosen such that there exist 
points at positions $i/(2k)$, for $i\in[0,2k]$ (assuming an arbitrary ``0'' point for the circle).
A configuration has the $k$ servers \emph{spaced uniformly} along the line/circle if there is a server at 
point $i/(2k)$ for all {\em odd} $i\in[0,2k]$. 

We define the \emph{best configuration}
\footnote{For a given metric, the best configurations may not exist; however, 
 for the circle and the line it does.} 
under $\omb$ as a configuration $C^*$ such that, for any other configuration $C$, 
$\minD(\cP^\delta_{C^*}[i]) \le \minD(\cP^\delta_C[i])$ for all points $i$. The following lemma defines the best configuration for the line and the circle.

\begin{lemma}
\label{lem:bestConf}
For the line and the circle, a configuration in which the servers are spaced uniformly along the metric space is the best configuration. 
\end{lemma}

\begin{proof}
We give the proof for the line (a very similar argument applies for the circle).
Let $C$ be the configuration in which all the servers are uniformly spaced along the line. 
That is, the furthest point on the line from the set of servers is at distance $1/2k$. Note that, under our assumptions about $\delta$ such a configuration always exists.
 By the definition of $C$, ignoring the first $k$ values of $0$, the values of $\minD(\cP^{\delta}_C[i])$ increase by $\delta$ every $2k$ steps. More formally, 
$$
\minD(\cP^{\delta}_C[i]) = \begin{cases}
0, &\text{for } 1 \le i \le k \\
\left\lceil\frac{i - k}{2k}\right\rceil \cdot \delta, &\text{for } k < i~.
\end{cases}
$$

Let $C'$ be a configuration of the servers that is not $C$. In the configuration of $C'$, there must be at least one 
point on the line with a cost higher than $1/2k$.
 By the definition of $C'$, ignoring the first $k$ values of $0$, the values of $\minD(\cP^{\delta}_{C'})$ increase by $\delta$ at most every $2k$ steps (i.e., $\minD(\cP^{\delta}_{C'}[\ell - 2k]) + \delta \le \minD(\cP^{\delta}_{C'}[\ell])$ for all $\ell > 3k$). 

Hence, there is a point $j$ such that, for every point $j'' < j$, $\minD(\cP^{\delta}_C[j'']) = \minD(\cP^{\delta}_{C'}[j''])$ and, for every point $j' \ge j$, $\minD(\cP^{\delta}_C[j']) \ge \minD(\cP^{\delta}_{C'}[j'])$.
Thus $C$ is a best configuration.

\end{proof}

Informally, in the following lemma, we compare the ``worst'' configuration to the best configuration. That is, we bound from above the distances to all points from a server with respect to any two configurations.

\begin{lemma}
\label{thm:biMatch}
  For any $\delta > 0$ and any two configurations $C_1$ and $C_2$ of $k$ servers on the line, $\minD(\cP^\delta_{C_1}[i]) \le 2k\minD(\cP^\delta_{C_2}[i])$ for every $i$. 
\end{lemma}

\begin{proof}
We make the natural assumption that in any configuration, there is no point that is occupied by more than one server.
We define $W$ as 
the \emph{worst configuration}, namely as the one that has the property that, for any configuration $C$, 
$\minD(\cP^\delta_{W}[i]) \ge \minD(\cP^\delta_C[i])$ for all points $i$ (as for the best configuration, a worst configuration may
not necessarily exist for every metric, but we will show that it exists for the line and the circle).
For the line, we claim that $W$ corresponds to all the servers being at one of ends of the line. In such a configuration, ignoring the first $k$ points at distance $0$, the values of $\minD(\cP^\delta_{W}[i])$ increase by $\delta$ at each step. More formally, 
$$\minD(\cP^\delta_{W}[i]) = \begin{cases}
  0, &\text{for } 1 \le i \le k \\
  (i - k)\delta, &\text{for } k < i~.
\end{cases}$$ 

For every other configuration $C$, $\minD(\cP^\delta_{W}[i]) \ge \minD(\cP^\delta_{C}[i])$.

Let $C^*$ be the configuration with the servers uniformly spaced along the line. By Lemma~\ref{lem:bestConf}, $C^*$ is the best configuration.

  Consider the ratio of $\minD(\cP_W^{\delta}[i])$ to
    $\minD(\cP_{C^*}^{\delta}[i])$. Modulo the initial $k$ points at distance $0$, the values of $\minD(\cP_W^{\delta}[i])$ increase by
    $\delta$ every step (i.e., $\minD(\cP^{\delta}_{W}[\ell - 1]) + \delta =  \minD(\cP^{\delta}_{W}[\ell])$ for all $\ell > k$) and the values of $\cP_{C^*}^{\delta}$ increase by
    $\delta$ every $2k$ steps (i.e., $\minD(\cP^{\delta}_{C^*}[\ell - 2k]) + \delta = \minD(\cP^{\delta}_{C^*}[\ell])$ for all $\ell > 3k$). This ratio is maximized when $i \ge 3k$ and
    $(i - k) \bmod{2k} = 0$, for which it attains a value of $2k$.
\end{proof}

From Lemma~\ref{thm:biMatch} and the notions of the proof, we also obtain the following.

\begin{lemma}
\label{cor:btw2servers}
For any $\delta > 0$ and any two configurations $C_1$ and $C_2$ of $k$ servers on the line or the circle, 
if $\cP^\delta_{C_1}[i]$ is  located between two adjacent servers of $C_1$, then 
$\minD(\cP^\delta_{C_1}[i]) \le k\minD(\cP^\delta_{C_2}[i])$.
\end{lemma}

\begin{proof}
From Theorem~\ref{thm:biMatch}, the best configuration $C^*$ places the servers uniformly along the line and, ignoring the first $k$ values at $0$, $\minD(\cP^{\delta}_{C^*}[i])$ increase by $\delta$ every $2k$ steps (i.e., $\minD(\cP^{\delta}_{C^*}[\ell - 2k]) + \delta = \minD(\cP^{\delta}_{C^*}[\ell])$ for all $\ell > 3k$). Hence, $\minD(\cP^{\delta}_{C^*}[i]) \ge \lceil (i-k)/2k \rceil \delta$ and the claim follows if $\minD(\cP^{\delta}_{C}[i]) \le \lceil (i-k)/2 \rceil \delta$, which we show in the following.
 
For any other configuration $C$, let $\cP^\delta_{C(s,t)} \subset \cP^\delta_{C}$ be the set of points between the servers $s$ and $t$ in configuration $C$, ordered by the distance to the nearest server. As the points are between two servers, the values of $\minD(\cP^{\delta}_{C(s,t)}[i])$ begin at $\delta$ and increase by $\delta$ every $2$ steps (i.e., $\minD(\cP^{\delta}_{C(s,t)}[1]) = \minD(\cP^{\delta}_{C(s,t)}[2]) = \delta$, and $\minD(\cP^{\delta}_{C(s,t)}[\ell - 2]) + \delta = \minD(\cP^{\delta}_{C(s,t)}[\ell])$ for all $\ell > 1$). 
Hence, as $\cP^\delta_{C(s,t)} \subset \cP^\delta_{C}$, it follows that $\minD(\cP^{\delta}_{C}[i]) \le \lceil (i-k)/2 \rceil \delta$.
\end{proof}

\subsubsection{Completing the analysis }
\label{subsec:greedy.line}
We will now use the bijection $\pi$ as defined explicitly in Definition~\ref{def:ob}, so as
to establish our upper bounds on the bijective ratio. An important observation is the following.
\begin{observation}\label{obs:costGdyKCnt}
  For any $\delta > 0$ and any server configuration $C$, the cost of $\gdy$ to serve a request at point $\cP^\delta_C[i]$ is $\minD(\cP^\delta_C[i])$ 
  and the cost of \kCtr \ is $2\minD(\cP^\delta_C[i])$.
\end{observation}

From its statement, \kCtr \ anchors its servers in the best configuration under \omb. 
By definition, to serve a request, the algorithm moves a server to a request and back to its original position. 
Hence, we obtain that, for any $\delta > 0$ and the best server configuration $C^*$, the cost of serving a request for \kCtr \  is
$2\minD(\cP^\delta_{C^*}[i])$. This immediately implies the following:
\begin{theorem}
  \kCtr \ has an asymptotic bijective ratio of at most $2$ for the $k$-server problem on the line and the circle.
  \label{thm:kcenter.2}
\end{theorem}
We will now use the framework of Lemma~\ref{lem:genApproach}, and we begin by applying it to the circle metric. 
Let $B$ be any online or offline algorithm. First note that by the definitions of \gdy \ and \kCtr, we have the following inequalities.
\begin{align}\label{eq:prop1} 
\gdy(\sigma[i] | B(\sigma[1,i-1])) &\le B(\sigma[i]) \\
\label{eq:propkCtr}
2\cdot \gdy(\sigma[i] | \kCtr(\sigma[1,i-1])) &\le \kCtr(\sigma[i]) 
\end{align} 

\begin{theorem}
\label{thm:circleK}
 \gdy \ has a bijective ratio of at most $k$ for the $k$-server problem on the circle.
 Moreover, \gdy \ has a bijective ratio of at most $k/2$ for the $k$-server problem on the circle against \kCtr.
\end{theorem}

\begin{proof}
We begin by proving the first part of the theorem, the second part follows along the same lines.

\medskip
\noindent{\emph{Part 1}:}
From~\eqref{eq:prop1}, we have that 
\begin{equation}\label{eq:Gdy1Circ}
  (k-1)B(\pi(\sigma)[i]) \ge (k-1)\gdy(\sigma[i] | B(\sigma[1,i-1])) ~.
\end{equation} 
In addition, from Lemma~\ref{cor:btw2servers}, it follows that 
\begin{equation}\label{eq:Gdy2Circ}
  k \cdot  \gdy(\pi(\sigma)[i] | B(\pi(\sigma)[1,i-1])) \ge \gdy(\sigma[i])
\end{equation}
since every request on the circle is located between two servers. Adding \eqref{eq:Gdy1Circ} and \eqref{eq:Gdy2Circ}, we thus obtain that, 
for any $\sigma$, 
$$(k-1)B(\pi(\sigma)[i]) \ge \gdy(\sigma[i]) - \gdy(\pi(\sigma)[i] | B(\pi(\sigma)[1,i-1])~.$$
Using this with~\eqref{eq:prop1} and applying Lemma~\ref{lem:genApproach}, we obtain the result.

\bigskip
\noindent{\emph{Part 2}:}
Inequality~\eqref{eq:propkCtr} implies 
\begin{equation}\label{eq:kCtr1Circ}
(k-1)\kCtr(\pi(\sigma)[i]) \ge 2(k-1) \cdot \gdy(\sigma[i] | \kCtr(\sigma[1,i-1]))~.
\end{equation}
From Lemma~\ref{cor:btw2servers} it follows that
\begin{equation}\label{eq:kCtr2Circ} 
  2k \cdot \gdy(\sigma[i] | \kCtr(\sigma[1,i-1])) \ge 2 \cdot \gdy(\sigma[i])
\end{equation}
since every request on the circle is located between two servers. 
Adding \eqref{eq:kCtr1Circ} and \eqref{eq:kCtr2Circ}, we thus obtain that, for any $\sigma$,
  $$(k-1)\kCtr(\pi(\sigma)[i]) \ge 2 (\gdy(\sigma[i]) - \gdy(\pi(\sigma)[i] | \kCtr(\pi(\sigma)[1,i-1])))~.$$
  Using this with~\eqref{eq:propkCtr}, and applying Lemma~\ref{lem:genApproach}, we obtain the result.
\end{proof}

We now move to the line metric. We first note that by combining~\eqref{eq:prop1} 
and Lemma~\ref{thm:biMatch} and by applying Lemma~\ref{lem:genApproach}, we obtain a {\em strict} 
bijective ratio of $2k$ for $\gdy$ on the line. We will also show a stronger, albeit asymptotic  
bound of $4k/3$, using amortized analysis based on Lemma~\ref{lem:amort}. 
First, we show, 
a general bound for $\gdy$ on the line as compared to some 
algorithm $B$. 
Later, we will choose appropriately the parameters in the statement of the lemma to compare $\gdy$ to an arbitrary algorithm and to \kCtr\ in particular. 
\begin{lemma}
\label{lem:genLineBetter}
    Let $B$ be any algorithm for the $k$-server problem on the line such that \\
$\gdy(\sigma[i] | B(\sigma[1,i-1])) \le d B(\sigma[i])$ for $d>0$. 
Suppose that there exist $c_1,c_2>0$ such that in the configuration of $\gdy$ right before $\sigma_i$, 
$$ \gdy(\sigma[i]) \le \begin{cases}
c_1 \gdy(\sigma[i] | B(\pi(\sigma)[1,i-1])), &\text{if  $\sigma[i]$ is between two servers} \\
c_2 \gdy(\sigma[i] | B(\pi(\sigma)[1,i-1])), &\text{otherwise.}
\end{cases}$$
Then, for any $\sigma$, $\gdy(\sigma) \le d \cdot \frac{2c_2c_1}{c_2 + c_1} \cdot B(\pi(\sigma)) + \eta~,$
where $\eta$ is a constant that depends on the diameter of the line. 
\end{lemma}

\begin{proof}
This proof makes use of a potential function argument. Define $\alpha$ to be equal to $\frac{c_2-c_1}{c_2+c_1}$. 
We define the potential function $\Phi := -\alpha \sum_{i=1}^{k-1} d(g_i,g_{i-1})$, i.e., $-\alpha$ times the sum of the distances between adjacent servers.  Let $a_i$ denote the amortized cost for $\sigma[i]$, i.e., $a_i = \gdy(\sigma[i]) + \Delta\Phi_i$, where $\Delta\Phi_i = \Phi_i - \Phi_{i-1}$. We distinguish between the following cases, concerning each request $\sigma[i]$. 
\begin{packed_item}
\item{\em Case 1: $\sigma[i]$ is between two non-outer-most servers.} \\
In this case, the change in potential is $\Delta\Phi_i = 0$ as the server that moves approaches one adjacent server by a distance of $B(\sigma[i])$ and moves away from its other adjacent server by the same distance. Hence, we have an amortized cost $$a_i = \gdy(\sigma[i]) \le c_1 \gdy(\sigma[i] | B(\pi(\sigma)[1,i-1]))~.$$
\item{\em Case 2: $\sigma[i]$  is between an end-point and an outer-most server.} \\
In this case, $\Delta\Phi_i = -\alpha\gdy(\sigma[i])$ and $$a_i = (1-\alpha) \gdy(\sigma[i]) \le c_2 (1-\alpha) \gdy(\sigma[i] | B(\pi(\sigma)[1,i-1]))~.$$
\item {\em Case 3: $\sigma[i]$ is between an outer-most server and its adjacent server.} \\
In this case, $\Delta\Phi_i = \alpha \gdy(\sigma[i])$ and $$a_i = (1 + \alpha) \gdy(\sigma[i]) \le c_1 (1+\alpha) \gdy(\sigma[i] | B(\pi(\sigma)[1,i-1]))~.$$
\end{packed_item}
Overall, for any $i$, $a_i \le c_1 (1+\alpha) \gdy(\sigma[i] | B(\pi(\sigma)[1,i-1]))$ as $c_2 (1 -\alpha) = c_1 (1+\alpha) \ge c_1$. The lemma follows by applying Lemma~\ref{lem:amort}.
\end{proof}

Using Lemma~\ref{lem:genLineBetter} and the properties of the $\omb$, we obtain the following bounds
for $\gdy$ against any algorithm (including offline), as well as $\kCtr$. 
\begin{theorem}
  Let $B$ be any algorithm for the $k$-server problem on the line. Then for any $\sigma$, $\gdy(\sigma) \le \frac{4k}{3} 
  B(\pi(\sigma)) + \eta ~,$ where $\eta$ is a constant. Moreover, 
 $\gdy(\sigma) \le \frac{2k}{3} \kCtr(\pi(\sigma)) + \eta$, where again $\eta$ is a constant. 
  \label{thm:4k.over.3}
\end{theorem}
\begin{proof}
We determine appropriate values for $d$, $c_1$, and $c_2$ in the statement of Lemma~\ref{lem:genLineBetter}. 
For the first part of the theorem, $d=1$ from~\eqref{eq:prop1}, $c_1 \le k$ from Lemma~\ref{cor:btw2servers},  and 
$c_2 \le 2k$ from Lemma~\ref{thm:biMatch}; hence  the first part of the theorem follows. For the second part of the theorem, 
$c_1$ and $c_2$ are as above, and, from~\eqref{eq:propkCtr}, we have $d = 1/2$.
\end{proof}

\section{The bijective ratio of the $k$-server problem on the star}
\label{sec:star}

In this section, we study the bijective ratio of the continuous $k$-server problem on the star. 
Here, a star consists of $m$ line segments
(called {\em rays}), not necessarily of the same length, 
which have a common origin called the {\em center}. 
One can think of such a metric as a transitional metric when moving from the line 
to trees that allows us to draw certain interesting conclusions concerning
the performance of algorithms under the bijective ratio. 
Similar to the line, we represent this metric by a {\em spider} graph, which
consists of a set of paths (of potentially different lengths) that intersect at the center, and in which all edges have the same length. 

Recall that on the line, Theorem~\ref{thm:kcenter.2} shows that \kCtr \ has a bijective ratio of 2,
whereas Theorems~\ref{thm:4k.over.3} and~\ref{lem:linek2} show that \gdy \ has a bijective ratio of $\Theta(k)$. 
In particular, our analysis of the bijective ratio of \kCtr \ matches its Max/Max ratio~\cite{BenBor94:maxmax}.
In contrast to our analysis of the bijective ratio of $\kCtr$ and $\gdy$, 
we show that on stars the bijective performance of these algorithms changes in a dramatic way. More precisely, in Theorem~\ref{thm:kctr.sucks.stars}, 
the bijective ratio of \kCtr \ is unbounded, while,
in Theorem~\ref{thm:starGdy}, we show that the bijective ratio of \gdy \ is at most $4k$. 
These results 
demonstrate that the bijective ratio is not only a generalization of Max/Max ratio, but it also classifies algorithms very differently
in terms of performance. 

\begin{theorem}
There exists a star $S$, and an online algorithm $A$ such that \kCtr \ has unbounded {\em asymptotic} bijective ratio against $A$ on $S$. 
\label{thm:kctr.sucks.stars}
\end{theorem}

\begin{proof}
Consider a star $S$ that consists of $m-1$ rays of length $d$ and one longer ray of length $4kd - d$. 
For this star, \kCtr \ anchors its servers on the long ray (see Figure~\ref{fig:kCtrSpider}). 
More specifically, the first server is placed at a distance $d$ from the center of the star with the remaining servers placed along the long 
ray with a spacing of $4d$ between them.
We also define an algorithm $A$ that anchors a server at the center of $S$, and the remaining $k-1$ servers
as in Figure~\ref{fig:kCtrSpider}. Similar to \kCtr, $A$ serves a request with the closest server, which then returns it to its anchor position. 
For this star, we show that there exist integers $j$ and $n$ such that the $j$-th cheapest sequence of \kCtr \ (among sequences in ${\cal I}_n$)
is at least $\Omega(d)$ times as costly as the $j$-th cheapest sequence of $A$. 

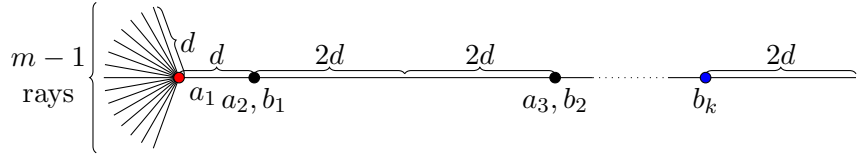
\begin{figure}[htb!]
  \centering
  \begin{tikzpicture}[scale=1]

    \foreach \a in {110,120,...,250}
      \draw (\a:0) -- (\a:1);

    \draw[decorate,decoration={brace}] (0.05,0) +(110:1) to node[midway,right] {$d$} (0.05,0);
    \draw[decorate,decoration={brace}] (-1.1,-1) to node[midway,left,align=center] {$m-1$ \\ rays} (-1.1,1);
    \draw[decorate,decoration={brace},yshift=0.05cm] (0,0) to node[midway,above] {$d$} (1,0);
    \draw[decorate,decoration={brace},yshift=0.05cm] (1,0) to node[midway,above] {$2d$} (3,0);
    \draw[decorate,decoration={brace},yshift=0.05cm] (3,0) to node[midway,above] {$2d$} (5,0);
    \draw[decorate,decoration={brace},yshift=0.05cm] (7,0) to node[midway,above] {$2d$} (9,0);

    \draw (0,0) -- (5.5,0);
    \draw[dotted] (5.5,0) -- (6.5,0);
    \draw (6.5,0) -- (9,0);
    \draw[fill=red] (0,0) node[below right] {$a_1$} circle (2pt);
    \draw[fill=black] (1,0) node[below] {$a_2,b_1$} circle (2pt);
    \draw[fill=black] (5,0) node[below] {$a_3,b_2$} circle (2pt);
    \draw[fill=blue] (7,0) node[below] {$b_k$} circle (2pt);

  \end{tikzpicture}
  \caption{An illustration of the lower bound construction for the \kCtr \ algorithm.  Here, we denote by $a_i$, $b_i$ the servers of $A$
  and \kCtr, respectively.}
  \label{fig:kCtrSpider}  
\end{figure}

We assume that the anchor position of the servers for algorithm~$A$ is the initial configuration of servers for the two algorithms. 

We  observe that after the initial anchoring of the servers for $A$ and $\kCtr$, each requested point always incurs the same cost for 
$A$, and always incurs the same cost for $\kCtr$. 

For some $\varphi \le d$ to be determined later, let $S^\kCtr_{\varphi,n}$ be the sequences of cost at most $2\varphi n$ for $\kCtr$. For a given sequence $\sigma \in S^\kCtr_{\varphi,n}$, let $\eps$ be the number of requests that have cost more than $2d$. Hence, $2\eps d \le 2\varphi n \iff \eps \le \frac{\varphi n}{d}$. Based on the anchor points for $\kCtr$, there are $2kd+k$ points with cost at most $2d$ and $N-2kd-k$ point with cost more than $2d$, where $N$ is the total number of nodes in the metric space.
We can bound from above the number of sequences in $S^\kCtr_{\varphi,n}$ as follows. (For this recall that the cost of a request is twice the distance to the nearest server since the server will be moved back to its original position afterwards.) Let 
\begin{align}
|S^\kCtr_{\varphi,n}| &\le \sum_{\eps=0}^{\frac{\varphi n}{d}} \binom{n}{\eps} (2kd+k)^{n-\eps}(N-2kd-k)^{\eps} \notag \\
             &\le 2^n \max_{\{0,\frac{\varphi n}{d}\}}\{(2kd+k)^{n-\eps}(md + 2kd)^{\eps}\} \notag \\
             &\le 2^n (2kd+k)^{n-\frac{\varphi n}{d}}(md + 2kd)^{\frac{\varphi n}{d}}~,\text{for $m > k$,} \notag \\
             &\le 2^n (3kd)^{n-\frac{\varphi n}{d}}(3md)^{\frac{\varphi n}{d}} \notag \\
             &\le (6kd)^{n}(3md)^{\frac{\varphi n}{d}} ~.\label{eq:ubSeqsK}
\end{align} 

Similarly, for some positive constant $c \ll d$, let $S^A_{c,n}$ be the sequences of cost at most $2cn$ for $A$. We can bound from below the number of sequences in $S^A_{c,n}$ as follows.
\begin{equation}\label{eq:lbSeqsA}
|S^A_{c,n}| \ge (c\cdot(m+2k-2)+k)^n \ge (cm)^n~.
\end{equation}

Setting $m = (kd)^3$ and $\varphi = d/3$ and using \eqref{eq:ubSeqsK}, we get that $|S^\kCtr_{d/3,n}| = 6^n3^{\frac{n}{3}}m^n$. For $c \ge 9$, $|S^A_{9,n}| > |S^\kCtr_{d/3,n}|$. Hence, for $j = |S^\kCtr_{d/3,n}|$, the bijective ratio for the $j$-th cheapest sequence is at least $d/27$.

This implies that the bijective ratio of \kCtr \ against $A$ is at least $\Omega(d)$ and, hence,
unbounded. 
\end{proof}

We define our bijection $\pi$ according to Definition~\ref{def:ob}.
Note that in the star, unlike the line and the circle, a best configuration, as defined in Section~\ref{subsec:ordered.bijection}
may not necessarily exist. However, the following corollary (which follows from Lemma~\ref{thm:biMatch}) shows that there exists a configuration that
is good enough. 

\begin{corollary}\label{cor:spiderGraph}
  Let $C$ be a configuration with a server at the centre. For any $\delta > 0$, $\cP^\delta_{C}[i] \le 2k\cP^\delta_{C'}[i]$, where $C'$ is any other configuration.
\end{corollary}

We note that, unlike the line, the approach of Lemma~\ref{lem:genApproach} cannot yield a bounded bijective ratio for \gdy. 
This is because the best and worst configurations (as defined in Section~\ref{subsec:ordered.bijection}) can be unbounded with respect to $\minD$. 
We will thus resort to amortized analysis, using a different potential function than the one used in Lemma~\ref{lem:genLineBetter}.

\begin{theorem}
\label{thm:starGdy}
    Let $B$ be any algorithm for the $k$-server problem on the uniform spider graph. 
    For any $\sigma$, $\gdy(\sigma) \le 4k B(\pi(\sigma)) + \eta$, where $\eta$ is a constant.
\end{theorem}

\begin{proof}
  The potential function is $\Phi := \sum_{i=1}^{k} d(c,g_i)$, where $c$ is the centre of the star  and $g_i$ is the $i$-th server. 
  Let $a_i$ denote the amortized cost for $\sigma[i]$, i.e., 
  $a_i = \gdy(\sigma[i]) + \Delta\Phi_i$, where $\Delta\Phi_i = \Phi_i - \Phi_{i-1}$. We distinguish three cases with respect to $\sigma[i]$.
\begin{packed_item}
\item {\em Case 1: A server serves $\sigma[i]$ away from the centre on the same ray.}
In this case, the change in potential is $\Delta\Phi_i = \gdy(\sigma[i])$. Hence, the amortized cost $a_i = 2\gdy(\sigma[i]) \le 2d(c,\sigma[i]) \le 4k \cdot \gdy(\pi(\sigma)[i] | B(\pi(\sigma)[1,i-1]))$, where the last inequality follows from Corollary~\ref{cor:spiderGraph}.

\item {\em Case 2: A server serves $\sigma[i]$ towards the centre on the same ray.}
In this case, the change in potential is $\Delta\Phi_i = -\gdy(\sigma[i])$. Hence, the amortized cost $a_i = 0$.

\item {\em Case 3: A server serves $\sigma[i]$ by moving a sever that lies on a different ray than the one of $\sigma[i]$.}
In this case, the change in potential is $\Delta\Phi_i = d(c,\sigma[i])-d(s_1,c)$. Hence, for the amortized cost we get $a_i = \gdy(\sigma[i]) + d(c,\sigma[i])-d(s_1,c) = 2d(c,\sigma[i]) \le 4k \cdot \gdy(\pi(\sigma)[i] | B(\pi(\sigma)[1,i-1]))$, where the last inequality follows from Corollary~\ref{cor:spiderGraph}.
\end{packed_item}
Thus, $a_i \le 4k \cdot \gdy(\sigma[i] | B(\pi(\sigma)[1,i-1])$ for all $i$. The theorem follows from~\eqref{eq:prop1} and Lemma~\ref{lem:amort}.
\end{proof}

\bibliographystyle{plain}
\bibliography{bijective}

\end{document}